\documentclass[journal,10pt]{IEEEtran}
\usepackage{cite}
\usepackage{graphicx}
\usepackage{psfrag}
\usepackage{subfigure}
\usepackage{url}
\usepackage[cmex10]{amsmath}
\usepackage{array}
\usepackage{graphics}
\usepackage{epsfig}
\usepackage{amsbsy}
\usepackage{amssymb}
\usepackage{amsthm}

\usepackage{color}

\usepackage{flushend}
\newtheorem{proposition}{Proposition}

\IEEEoverridecommandlockouts
\begin{document}

\title{Secure Communication Via a Wireless Energy Harvesting Untrusted Relay}

\author{Sanket S. Kalamkar and~Adrish Banerjee
\thanks{The authors are with the Department of Electrical Engineering, Indian Institute of Technology Kanpur, India. (e-mail:  $\lbrace$kalamkar, adrish$\rbrace$@iitk.ac.in).}
}

\maketitle  

\begin{abstract}
The broadcast nature of the wireless medium allows unintended users to eavesdrop the confidential information transmission. In this regard, we investigate the problem of secure communication between a source and a destination via a wireless energy harvesting untrusted node which acts as a helper to relay the information; however, the source and destination nodes wish to keep the information confidential from the relay node. To realize the positive secrecy rate, we use destination-assisted jamming. Being an energy-starved node, the untrusted relay harvests energy from the received radio frequency signals, which include the source's information signal and the destination's jamming signal. Thus, we utilize the jamming signal efficiently by leveraging it as a useful energy source. At the relay, to enable energy harvesting and information processing, we adopt power splitting (PS) and time switching (TS) policies. To evaluate the secrecy performance of this proposed scenario, we derive analytical expressions for two important metrics, viz., the secrecy outage probability and the ergodic secrecy rate. The numerical analysis reveals the design insights into the effects of different system parameters like power splitting ratio, energy harvesting time, target secrecy rate, transmit signal-to-noise ratio (SNR), relay location, and energy conversion efficiency factor, on the secrecy performance. Specifically, the PS policy achieves better optimal secrecy outage probability and optimal ergodic secrecy rate than that of the TS policy at higher target secrecy rate and transmit SNR, respectively. 
\end{abstract}
\begin{IEEEkeywords}
Destination-assisted jamming, ergodic secrecy rate, secrecy outage probability, untrusted relay, wireless energy harvesting.
\end{IEEEkeywords}
\IEEEpeerreviewmaketitle
\section{Introduction}
\subsection{Wireless Energy Harvesting and Cooperative Relaying}
\PARstart{E}{nergy} harvesting is a popular source of energy to power wireless devices~\cite{paradiso,visser,lu}. It holds the potential to prolong the lifetime of energy-constrained nodes and simultaneously avoids the frequent recharging and replacement of batteries, which otherwise would be inconvenient or unacceptable (e.g., medical devices implanted inside a human body). Besides harvesting energy from natural sources like solar, thermal, and wind, the radio frequency (RF) signals in the surrounding wireless environment is a viable source of energy. Exploiting that RF signals can carry both energy and information together, \cite{varshney,grover,rui2} have proposed the simultaneous wireless energy harvesting and information transfer from the same received RF signals. Since it is difficult for a receiver to harvest energy and process information from the same signal, two practical policies for the wireless energy harvesting and information processing are proposed in~\cite{rui2,rui4,nasir}. One policy is the power splitting policy where the receiver splits the received power between energy harvesting and information processing. The second policy involves time switching which divides the time between energy harvesting and information processing.

Such simultaneous energy harvesting and information processing has an application in cooperative relaying~\cite{nasir,ishibashi,krik3,poor,nasir2,hechen,ding_tvt,sanket,gu1}. Using the broadcast nature of the wireless medium, the source transmits the information to an intermediate node, that retransmits it to the destination. In this setup, the relay harvests energy from the received RF information signal and uses it further to forward the information to the destination. The energy harvesting along with the information transfer can prolong the lifetime of a relay, which in turn, facilitates the information cooperation.

\subsection{Physical-Layer Security and Untrusted Relaying}
Although the broadcast nature of the wireless medium has facilitated the cooperative communication, it has also allowed unintended nodes, also known as eavesdroppers, to hear the confidential information transmission between the source and the intended destination via a relay, leading to the insecure communication. Traditional approaches to achieve the secure communication include upper-layer cryptographic techniques which require intensive key distribution and management. Unlike this paradigm, the physical-layer information-theoretic security achieves the secure communication by exploiting the nature of the wireless channel. In this regard, Wyner introduced the wiretap channel and showed that the perfect secure communication was possible without relying on private keys~\cite{wyner}.

As to the cooperative relaying, the works in~\cite{krik1,ding1,liu,vo,ding3,dong,jeong2,liu6} have studied the physical-layer security in the presence of external eavesdroppers that are different from the relay and try to intercept the source-relay and/or relay-destination communications. Even in the absence of external eavesdroppers, the secure communication between source and destination may still be a concern, as one may wish to keep the source-destination communication secret from the relay itself despite its cooperation in forwarding the information~\cite{oohama}.
In this case, the relay is considered as an eavesdropper. The model of untrusted relay has practical applications in defence, financial, and government networks, where all users do not have the same rights to access the information. Also, if the relay belongs to a different network, it may not have the privilege to access the data as that of the source and the destination. 

In~\cite{he}, the authors show that even the communication via an untrusted relay can be beneficial for the relay channel with orthogonal components. The works in~\cite{he3,he1} show that the positive secrecy rate is achievable with the destination-assisted jamming, where the destination sends jamming signals during the source-relay communication. The references~\cite{li1,li2,wang,liu1,ju,park2015} investigate the information-theoretic security performance for amplify-and-forward (AF) relays under the fading channel and destination-assisted jamming. The work in~\cite{huang} studies the secrecy outage probability performance of the communication via an untrusted multi-antenna relay. 
In~\cite{zhang1}, the authors advocate the use of friendly jammers to secure the communication via an untrusted relay. To achieve secure as well as spectral efficient communication, the authors in \cite{khoda} propose link adaptation and relay assignment. With distributed beamforming and opportunistic relaying, the reference~\cite{kim} studies the capacity scaling and diversity order for secure relaying. In~\cite{ryu}, the authors examine the secure relay-assisted communication, where legitimate users, rather considering the relay completely untrusted, have a degree of trust about the relay.\vspace*{-2mm}

\subsection{Wireless Energy Harvesting with Physical-Layer Security}
Recently, with wireless energy harvesting, a few works have studied the physical-layer security in the presence of external eavesdroppers for different scenarios like point-to-point communication with a single antenna~\cite{xing,biao} and multiple antennas~\cite{schober,khandaker,feng,shi}, and the cooperative communication via a relay~\cite{quan,xing2,xing2015,chen2014}. In~\cite{quan}, in the presence of the external energy harvesting receiver, the authors study the secure relay beamforming with simultaneous wireless information and energy transfer. The work in~\cite{xing2} investigates the secrecy performance for an AF relay wiretap channel when the external helpers, who act as jammers to the eavesdropper, harvest energy from the source's transmission. In the presence of an external eavesdropper, in~\cite{xing2015}, authors have studied the secure communication between a source and a destination via multiple energy harvesting relays; while the work in~\cite{chen2014} investigates the secrecy performance of the source-destination communication via an energy harvesting relay with multiple antennas. However, the works in~\cite{quan,xing2,xing2015,chen2014} assume the relay to be trusted, and external eavesdroppers attempt to intercept the relay-assisted communication between the source and the destination. Also, the works on untrusted relay till now have assumed that the conventional energy source, such as battery, powers the relay (see, e.g., \cite{oohama,he,he3,he1,li1,li2,wang,liu1,ju,huang,zhang1,khoda,kim,ryu}).

In this work, we address the problem of secure communication via an energy harvesting untrusted relay. When an untrusted relay harvests energy from the received RF signals, the jamming signal can act as a potential energy source besides its original purpose of realizing the secure communication via untrusted relay. This allows us to use the jamming signal efficiently.\vspace*{-2mm}

\subsection{Contributions and Key Results}
In this paper, we investigate the secrecy performance of a two-hop communication between a source and a destination, where the source uses an AF wireless energy harvesting untrusted relay to forward the confidential information to the legitimate destination. To keep the information secret from the relay, we consider the destination-assisted jamming. The relay harvests energy from received RF signals, which include the information signal from the source and the jamming signal from the destination. We use power splitting (PS) and time switching (TS) receiver architectures~\cite{rui2} at the relay to facilitate the energy harvesting and information processing. We summarize the main contributions and key results below.
\begin{itemize}
\item With the jamming signal leveraged as a useful energy source under both PS and TS policies, for an energy harvesting AF relay, we derive analytical expressions for two important measures of secrecy performance---the secrecy outage probability and the ergodic secrecy rate.
\item We further compare PS and TS policies, where we show that, at higher target secrecy rate and transmit signal-to-noise ratio (SNR), PS policy achieves lower secrecy outage probability and higher ergodic secrecy rate, respectively.
\item The numerical results also show that the power splitting ratio in PS policy and energy harvesting time in TS policy have both constructive and destructive effects on the secure communication between source and destination. Thus, the optimal power splitting ratio in PS policy and the optimal energy harvesting time in TS policy that maximize the secrecy performance do exist. 
\item In addition, for both PS and TS policies, the numerical analysis shows that, the optimal secrecy performance is achieved when the relay is located closer to the destination than to the source. This is in contrast with the case where the wireless energy harvesting relay is considered to be trusted, and the optimum location of the relay is closer to the source. 
\end{itemize}\vspace*{-2mm}
\begin{figure}
\centering
\includegraphics[scale=0.33]{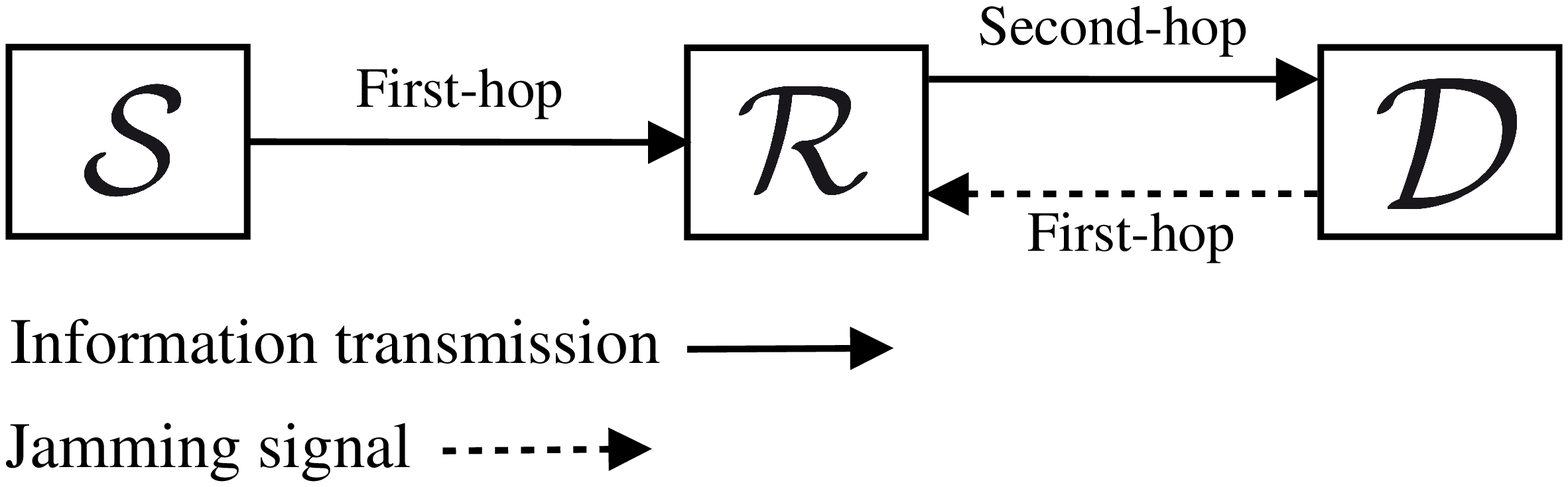}\vspace*{-2mm}
\caption{System Model for the secure communication between a source ($\mathcal{S}$) and a destination ($\mathcal{D}$) via an energy harvesting untrusted relay ($\mathcal{R}$) with destination-assisted jamming.}\vspace*{-4mm}
\label{fig:sys}
\end{figure}
\subsection{Organization of the Paper}
We structure the rest of the paper as follows. Section~\ref{sec:sys} describes the system model for the two-hop secure communication via an energy harvesting untrusted relay using the destination-assisted jamming. In Sections~\ref{sec:PS} and \ref{sec:TS}, utilizing the jamming signal as a useful energy source, we derive analytical expressions for the secrecy outage probability and the ergodic secrecy rate for PS policy and TS policy based relaying. We present numerical results in Section~\ref{sec:results}, where we also discuss the effects of different system parameters on the secrecy performance of the relay-assisted communication and obtain various design insights. Finally, we provide concluding remarks in Section~\ref{sec:conc}.\vspace*{-1mm}

\section{System Model}
\label{sec:sys}

\subsection{Destination-Assisted Jamming and Channel Model}
As shown in Fig.~\ref{fig:sys}, a source ($\mathcal{S}$) communicates with a destination ($\mathcal{D}$) via an amplify-and-forward (AF) energy harvesting relay ($\mathcal{R}$). Despite relay's information cooperation, the source and destination nodes wish to keep the information secret from the relay. To maintain the confidentiality of the source information, the destination sends a jamming signal to the relay when source transmits the information to the relay. Each node operates in a half-duplex mode and has a single antenna. The direct link between $\mathcal{S}$ and $\mathcal{D}$ is unavailable.\footnote{Since the destination operates in a half-duplex mode and sends the jamming signal to the relay during the source's transmission, it cannot receive the information from the source.} Let us denote the coefficient of the channel between nodes $i$ and $j$ by $h_{ij}$. We consider a quasi-static block-fading Rayleigh channel between two nodes, as in~\cite{nasir,nasir2,li1,li2}. That is, the channel remains constant over a slot-duration of $T$ during which $\mathcal{S}$ transmits to $\mathcal{D}$ via $\mathcal{R}$ and changes independently from one slot to another. The channel power gain is given by $|h_{ij}|^{2}$, which has exponential distribution with mean $\lambda_{ij}$, i.e.,
\begin{equation}
f_{|h_{ij}|^2}(x) = \frac{1}{\lambda_{ij}}\exp\left({-\dfrac{x}{\lambda_{ij}}}\right), \,\,\, x \geq 0,
\end{equation}
where $f_{|h_{ij}|^2}(x)$ is the probability density function of random variable $|h_{ij}|^2$. We assume the channel between $\mathcal{R}$ and $\mathcal{D}$ reciprocal, as in~\cite{li1,huang,li2,wang,liu1}, i.e., $h_{\mathcal{R}\mathcal{D}} = h_{\mathcal{D}\mathcal{R}}$. In this work, the source is assumed to have no channel state information (CSI), while the CSI of $\mathcal{S}-\mathcal{R}$ and $\mathcal{R}-\mathcal{D}$ channels are available at the relay and destination, respectively~\cite{nasir, nasir2,ishibashi,krik3}.

\subsection{Energy Harvesting and Information Processing Model}
\label{sec:energy_act}
The untrusted relay harvests energy from the received RF signals which it uses to forward the source's information to the destination. To activate the energy harvesting circuitry at the relay, the received power must exceed the minimum threshold power $\theta_H$~\cite{lu,guo2015,duong2015}.\footnote{The threshold $\theta_H$ is usually between $-30~\mathrm{dBm}$ to $-10~\mathrm{dBm}$, depending on various factors like channel conditions, frequency of the received RF signals, and
energy harvesting circuitry type (linear, non-linear, tunable, etc.)\cite{lu}.} We assume that the relay has no other energy source and uses the harvested energy completely for the transmission as the power consumed by the relay's transmit/receive circuitry is negligible compared to the power required for the transmission~\cite{nasir,nasir2}.

We adopt following two different receiver architectures based policies at the relay to separately harvest energy from the received RF signals and process the information~\cite{rui4}. 
\begin{itemize}
\item[1.] Power splitting (PS) policy: The relay uses a part of the received power to harvest the energy and the remaining part for the information processing. 
\item[2.] Time switching (TS) policy: The relay switches between the energy harvesting and the information processing. That is, the relay uses a fraction of the time of a slot to harvest the energy and the remaining time for the information processing and relaying.
\end{itemize}
Note that the relay may attempt to decode the source information with the power used for the information processing.
\section{Power Splitting Policy Based Relaying}
\label{sec:PS}
\begin{figure}
\centering
\includegraphics[scale=0.16]{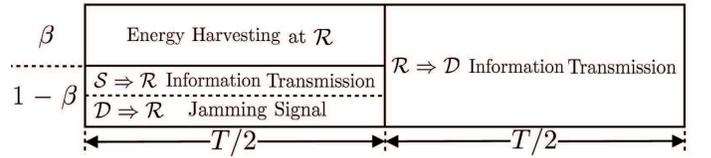}
\caption{Power splitting policy for the secure communication via an energy harvesting untrusted relay.}
\label{fig:PSP}
\end{figure}
Fig.~{\ref{fig:PSP}} shows the power splitting (PS) policy based relaying protocol, where the source-to-destination communication happens in a slot of duration $T$. Two phases of equal duration $T/2$ divide the slot. In the first phase, the source transmits information to the relay with power $P_{\mathcal{S}}$. At the same time, the destination sends a jamming signal with power $P_{\mathcal{D}}$ to the relay to maintain the confidentiality of the source information from the relay. The relay uses a fraction $\beta$ of the received power for energy harvesting and the remaining $(1-\beta)$ portion for information processing, where $0 \leq \beta \leq 1$. Using the harvested energy, in the second phase, the relay forwards the received information to destination after amplification.

\subsection{Energy Harvesting at Relay}
In the aforementioned PS policy, the relay harvests energy $E_{H}$ given as
\begin{equation}
E_{H} = \eta \beta \left(P_{\mathcal{S}} |h_{\mathcal{S}\mathcal{R}}|^2 + P_{\mathcal{D}} |h_{\mathcal{D}\mathcal{R}}|^2\right)(T/2),
\label{eq:harv_energy_ps}
\end{equation}
where $\eta$ is the energy conversion efficiency factor with $0 < \eta \leq 1$, which is dependent on the energy harvesting circuitry of the relay. The terms $P_{\mathcal{S}} |h_{\mathcal{S}\mathcal{R}}|^2$ and $P_{\mathcal{D}} |h_{\mathcal{D}\mathcal{R}}|^2$ in \eqref{eq:harv_energy_ps} denote the power received at the relay due to the information signal from the source and the jamming signal from the destination, respectively. In the second phase of duration $T/2$, the relay's transmit power to forward the information to destination is given as
\begin{equation}
P_{H} = \frac{E_{H}}{T/2} =  \eta \beta \left(P_{\mathcal{S}} |h_{\mathcal{S}\mathcal{R}}|^2 + P_{\mathcal{D}} |h_{\mathcal{D}\mathcal{R}}|^2\right).
\label{eq:harv_power_ps}
\end{equation}
\subsection{Information Processing and Relaying Protocol}
In phase one, the received signal $y_{\mathcal{R}}$ for the information processing at the relay is given by
\begin{equation}
y_{\mathcal{R}} = \sqrt{(1-\beta)P_{\mathcal{S}}}h_{\mathcal{S}\mathcal{R}}x_{\mathcal{S}} +\sqrt{(1-\beta)P_{\mathcal{D}}}h_{\mathcal{D}\mathcal{R}}x_{\mathcal{D}} + n_{\mathcal{R}},
\label{eq:rec_rel}
\end{equation}
where $x_{\mathcal{S}}$ is the source message with unit power, $x_{\mathcal{D}}$ is the unit power jamming signal sent by the destination, and $n_{\mathcal{R}}$ is the additive white Gaussian noise (AWGN) at the relay. We assume that the power splitting does not affect the noise power~\cite{poor,chua}. Based on the received signal $y_{\mathcal{R}}$ in \eqref{eq:rec_rel}, the relay may attempt to decode the source message $x_{\mathcal{S}}$. We can write the signal-to-noise ratio (SNR) at the relay as
\begin{equation}
\gamma_{\mathcal{R}} = \frac{(1-\beta)P_{\mathcal{S}}|h_{\mathcal{S}\mathcal{R}}|^{2}}{(1-\beta)P_{\mathcal{D}}|h_{\mathcal{D}\mathcal{R}}|^{2} + N_0},
\label{eq:snr_rel}
\end{equation}
where $N_0$ is the noise power of AWGN $n_{\mathcal{R}}$.

In phase two, the relay amplifies the received signal $y_{\mathcal{R}}$ by a factor $\xi$ based on its power constraint and forwards the resultant signal $x_{\mathcal{R}}$ to the destination, which is  given as
\begin{align}
x_{\mathcal{R}} &= \xi y_{\mathcal{R}} \label{eq:ampl11} \\
& =\sqrt{\frac{P_{H}}{(1-\beta)P_{\mathcal{S}}|h_{\mathcal{S}\mathcal{R}}|^{2}+(1-\beta)P_{\mathcal{D}}|h_{\mathcal{D}\mathcal{R}}|^{2} + N_0}}y_{\mathcal{R}}.
\label{eq:ampl}
\end{align} 
Then, we substitute~\eqref{eq:rec_rel} in~\eqref{eq:ampl11} and then use~\eqref{eq:ampl11} to write the received signal $y'_{\mathcal{D}}$ at the destination as
\begin{align}
y'_{\mathcal{D}} &= h_{\mathcal{R}\mathcal{D}}x_{\mathcal{R}} + n_{\mathcal{D}} \nonumber\\
& = \xi\sqrt{(1-\beta)P_{\mathcal{S}}}h_{\mathcal{S}\mathcal{R}}h_{\mathcal{R}\mathcal{D}}x_{\mathcal{S}} \nonumber \\
& + \xi\sqrt{(1-\beta)P_{\mathcal{D}}}h_{\mathcal{R}\mathcal{D}}h_{\mathcal{D}\mathcal{R}}x_{\mathcal{D}} 
+ \xi h_{\mathcal{R}\mathcal{D}} n_{\mathcal{R}} + n_{\mathcal{D}},
\label{eq:rec_des}
\end{align}
where $n_{\mathcal{D}}$ is the AWGN at the destination with power $N_0$. Since $x_{\mathcal{D}}$ is the jamming signal sent by the destination itself to the relay in phase one, the destination can remove the term  $\xi\sqrt{(1-\beta)P_{\mathcal{D}}}h_{\mathcal{R}\mathcal{D}}h_{\mathcal{D}\mathcal{R}}x_{\mathcal{D}}$ from \eqref{eq:rec_des} and decode the source information from the rest of the received signal.\footnote{In the case of channel estimation errors, the destination will have inaccurate knowledge of the channel gain on the relay-destination link, due to which it will not be able to cancel the jamming signal completely, causing self-interference. This, in turn, will reduce the received SNR at the destination, deteriorating the secrecy performance. Given the scope of this paper is to analyze the untrusted nature of an energy harvesting relay on the source-destination communication, we restrict ourselves to study the secrecy performance without channel estimation errors.} Thus, the resultant received signal $y_{\mathcal{D}}$ at the destination becomes
\begin{equation}
y_{\mathcal{D}} = \xi\sqrt{(1-\beta)P_{\mathcal{S}}}h_{\mathcal{S}\mathcal{R}}h_{\mathcal{R}\mathcal{D}}x_{\mathcal{S}} + \xi h_{\mathcal{R}\mathcal{D}} n_{\mathcal{R}} + n_{\mathcal{D}}.
\label{eq:rec_des1}
\end{equation} 
Finally, substituting $P_{H}$ from \eqref{eq:harv_power_ps} in~\eqref{eq:ampl}, and then using $\xi$ from~\eqref{eq:ampl} in~\eqref{eq:rec_des1}, we get
\begin{align}
y_{\mathcal{D}} &= \frac{\sqrt{\eta \beta(1-\beta)P_{\mathcal{S}} \left(P_{\mathcal{S}} |h_{\mathcal{S}\mathcal{R}}|^2 + P_{\mathcal{D}} |h_{\mathcal{D}\mathcal{R}}|^2\right)}h_{\mathcal{S}\mathcal{R}}h_{\mathcal{R}\mathcal{D}}x_{\mathcal{S}}}{\sqrt{(1-\beta)P_{\mathcal{S}}|h_{\mathcal{S}\mathcal{R}}|^{2}+(1-\beta)P_{\mathcal{D}}|h_{\mathcal{D}\mathcal{R}}|^{2} + N_0}} \nonumber \\
& + \frac{\sqrt{\eta \beta \left(P_{\mathcal{S}} |h_{\mathcal{S}\mathcal{R}}|^2 + P_{\mathcal{D}} |h_{\mathcal{D}\mathcal{R}}|^2\right)}h_{\mathcal{R}\mathcal{D}}n_{\mathcal{R}}}{\sqrt{(1-\beta)P_{\mathcal{S}}|h_{\mathcal{S}\mathcal{R}}|^{2}+(1-\beta)P_{\mathcal{D}}|h_{\mathcal{D}\mathcal{R}}|^{2} + N_0}} + n_{\mathcal{D}}.
\label{eq:rec_des2}
\end{align}
The first term on the right hand side of \eqref{eq:rec_des2} represents the signal part, while the second and third terms correspond to the total received noise at the destination. Then, the SNR at the destination can be written as
\begin{align}
\gamma_{\mathcal{D}} = \frac{\eta \beta(1-\beta)P_{\mathcal{S}} |h_{\mathcal{S}\mathcal{R}}|^{2}|h_{\mathcal{R}\mathcal{D}}|^{2}}{ \eta \beta |h_{\mathcal{R}\mathcal{D}}|^2N_0 + N_0 (1-\beta) + \frac{N_0^2}{\left(P_{\mathcal{S}} |h_{\mathcal{S}\mathcal{R}}|^2 + P_{\mathcal{D}} |h_{\mathcal{D}\mathcal{R}}|^2\right)}}.
\label{eq:snr_des}
\end{align}

\subsection{Secure Communication via an Untrusted Relay}
When the relay is considered untrusted, we can write the instantaneous secrecy rate $R_{\mathrm{sec}}$ of the relay-assisted communication as~\cite{bloch}
\begin{align}
R_{\mathrm{sec}} &= \frac{1}{2}\left[\log_2\left(1 +\gamma_{\mathcal{D}}\right) - \log_2\left(1 +\gamma_{\mathcal{R}}\right)\right]^{+} \nonumber \\
&=\frac{1}{2}\left[\log_2\left(\frac{1 + \gamma_{\mathcal{D}}}{1 + \gamma_{\mathcal{R}}}\right)\right]^{+},
\label{eq:sec_cap}
\end{align}
where $[x]^{+} = \max(x, 0)$. The factor $\frac{1}{2}$ represents the effective communication time between the source and the destination. For the rest of the Section~\ref{sec:PS}, we assume $P_{\mathcal{S}} = P_{\mathcal{D}} = P$ for analytical tractability.

\subsubsection{Secrecy Outage Probability}
The secrecy outage probability is an important measure of the secrecy performance. It allows us to determine the probability of attaining a target secrecy rate.
Given the energy harvesting circuitry of the relay is active, we can express the secrecy outage probability as~\cite{bloch}
\begin{align}
P_{\mathrm{out}} = \mathbb{P}\left( R_{\mathrm{sec}} < R_{\mathrm{th}}\right),
\label{eq:outt}
\end{align}
where $\mathbb{P}(\cdot)$ denotes the probability, $R_{\mathrm{sec}}$ is the instantaneous secrecy rate given by \eqref{eq:sec_cap}, and $R_{\mathrm{th}}$ is the target secrecy rate. Then, substituting $\gamma_{\mathcal{R}}$ from \eqref{eq:snr_rel} and $\gamma_{\mathcal{D}}$ from \eqref{eq:snr_des}, we can rewrite \eqref{eq:outt} as
\vspace*{-3mm}

{{\small
\begin{align}
\! P_{\mathrm{out}} = \mathbb{P}\!  \left(\!\!\frac{1 +  \frac{\eta \beta(1-\beta)P|h_{\mathcal{S}\mathcal{R}}|^{2}|h_{\mathcal{R}\mathcal{D}}|^{2}}{\left(N_0\eta \beta |h_{\mathcal{R}\mathcal{D}}|^{2} + N_0(1-\beta)\right) + \frac{N_0^2}{P(|h_{\mathcal{S}\mathcal{R}}|^{2} + |h_{\mathcal{R}\mathcal{D}}|^{2})}}}{1+ \frac{(1-\beta)P|h_{\mathcal{S}\mathcal{R}}|^{2}}{(1-\beta)P|h_{\mathcal{R}\mathcal{D}}|^{2} + N_0}} <  2^{2R_{\mathrm{th}}}\!\!\right)\!\!\!.
\label{eq:out1}
\end{align}}}\vspace*{-3mm}

\noindent We can further express the secrecy outage probability in~\eqref{eq:out1} analytically as given in Proposition~\ref{prop:out_ex}.
\begin{proposition}
\label{prop:out_ex}
The secrecy outage probability for PS policy can be approximately expressed as

\begin{align}
P_{\mathrm{out}} &\approx 1 - \frac{1}{\lambda_{\mathcal{RD}}}\int_{\theta_1}^{\infty}\exp\left(-\frac{\delta-1}{\nu(x)\lambda_{\mathcal{SR}}} - \frac{x}{\lambda_{\mathcal{RD}}}\right)\,\mathrm{d}x, 
\label{eq:out_fin}
\end{align}
where $\delta = 2^{2R_{\mathrm{th}}}$ with
\begin{subequations}
\begin{equation}
\theta_1 = \frac{\frac{\delta-1}{1-\beta} + \sqrt{\left(\frac{\delta-1}{1-\beta}\right)^{2} + \frac{4 \delta P}{\eta \beta N_0}}}{2(P/N_0)},
\end{equation} and
\begin{equation}
\nu(x) = (1-\beta)\left(\frac{\eta \beta P x}{N_0 \left(\eta \beta x + (1-\beta)\right)}-\frac{P \delta}{P(1-\beta)x + N_0}\right).
\end{equation} 
\end{subequations}
\end{proposition}
\begin{proof}
See Appendix~\ref{app:out_ex}.
\end{proof}
Equation~\eqref{eq:out_fin} is obtained using the high SNR approximation of the received SNR at the destination given as
\begin{equation}
\gamma_{\mathcal{D}} \approx \frac{\eta \beta(1-\beta)P|h_{\mathcal{S}\mathcal{R}}|^{2}|h_{\mathcal{R}\mathcal{D}}|^{2}}{N_0\left(\eta \beta |h_{\mathcal{R}\mathcal{D}}|^{2} + (1-\beta)\right)}.
\label{eq:snr_des_appr}
\end{equation}
Equation~\eqref{eq:snr_des_appr} can be obtained from the exact expression given in \eqref{eq:snr_des} of $\gamma_{\mathcal{D}}$ by neglecting the term $\frac{N_0^2}{\left(P_{\mathcal{S}} |h_{\mathcal{S}\mathcal{R}}|^2 + P_{\mathcal{D}} |h_{\mathcal{D}\mathcal{R}}|^2\right)}$ (due to negligible $N_0^2$ at high SNR) from the denominator of \eqref{eq:snr_des}. The approximation in \eqref{eq:snr_des_appr} is analytically more tractable than the exact expression in \eqref{eq:out1}.\footnote{In fact, the complex structure of~\eqref{eq:out1} does not allow us to get an exact analytical expression for the secrecy outage probability. This is because, the term $\left( |h_{\mathcal{S}\mathcal{R}}|^2 +  |h_{\mathcal{D}\mathcal{R}}|^2\right)$ in the denominator of~\eqref{eq:out1} prevents the separation of two random variables $|h_{\mathcal{S}\mathcal{R}}|^2$ and $|h_{\mathcal{R}\mathcal{D}}|^2$, which in turn, impedes the simplification of \eqref{eq:out1} to get an exact analytical expression.}
Although the integral in~\eqref{eq:out_fin} cannot be expressed in a closed form, it can be easily evaluated numerically as the integrand consists of elementary functions.

As aforementioned in Section~\ref{sec:energy_act}, the received power at the relay must be greater than the minimum power threshold $\theta_H$ to activate the energy harvesting circuitry. Using channel reciprocity on the relay-destination link, we can write the received power $P_{R}$ at the relay as
\begin{equation}
P_{R} = \left(P |h_{\mathcal{S}\mathcal{R}}|^2 + P |h_{\mathcal{R}\mathcal{D}}|^2\right).
\end{equation}
If the received power $P_{R}$ is less than the power threshold $\theta_H$, the energy harvesting circuitry at the relay stays inactive, leading to the power outage. The following proposition gives the expression for the power outage probability $\mathbb{P}\left(P_R < \theta_H\right)$.
\begin{proposition}
\label{prop:eout}
We write the power outage probability $P_{\mathrm{p,out}}$ as follows:
\begin{equation}
\label{eq:eout}
P_{\mathrm{p,out}} =  \left\{
  \begin{array}{l l}
    1 - \frac{\lambda_{\mathcal{SR}}}{\lambda_{\mathcal{SR}} - \lambda_{\mathcal{RD}}}\exp\left(-\frac{\theta_H}{P \lambda_{\mathcal{SR}} }\right) \\
    -  \frac{\lambda_{\mathcal{RD}}}{\lambda_{\mathcal{RD}} - \lambda_{\mathcal{SR}}}\exp\left(-\frac{\theta_H}{P \lambda_{\mathcal{RD}} }\right), & \quad \mathrm{if}\,\, \lambda_{\mathcal{SR}} \neq  \lambda_{\mathcal{RD}} \\
    \Upsilon\left(2, \frac{\theta_H}{P\lambda_{\mathcal{SR}}} \right), & \quad \mathrm{if}\,\, \lambda_{\mathcal{SR}} = \lambda_{\mathcal{RD}},\\
  \end{array} \right.
  \end{equation}
where $\Upsilon(a,t) = \int_{0}^{t} x^{a-1} \exp(-x)\mathrm{d}x$ is the lower incomplete Gamma function.
\end{proposition}
\begin{proof}
See Appendix~\ref{app:eout}.
\end{proof}
For an energy constrained untrusted relay, a secrecy outage can also occur if the power received by the relay is insufficient to activate the energy harvesting circuitry~\cite{duong2015}. Thus, combining with \eqref{eq:out_fin}, we can write the overall secrecy outage probability $P_{\mathrm{out}}^{\mathrm{s}}$ as~\cite{duong2015}
\begin{equation}
P_{\mathrm{out}}^{\mathrm{s}} = P_{\mathrm{p,out}} + (1- P_{\mathrm{p,out}})P_{\mathrm{out}},
\label{eq:tot_out}
\end{equation}
where $P_{\mathrm{out}}$ is given by \eqref{eq:out_fin}.
\subsubsection{Probability of Positive Secrecy Rate}
The destination-assisted jamming helps to keep the source information confidential from the relay and achieve the secure communication. In this regard, the probability $P_{\mathrm{pos}}$ of achieving strictly positive secrecy rate is an important measure of the secrecy performance. We provide the exact and approximate analytical expression for $P_{\mathrm{pos}}$ in the following proposition.
\begin{proposition}
\label{prop:ppos}
We write the exact and high SNR approximation analytical expressions for the probability of achieving strictly positive secrecy rate $P_{\mathrm{pos}}$ as follows:
\begin{subequations}
\begin{align}
P_{\mathrm{pos}} &= (1 - P_{\mathrm{p,out}})\Bigg[\exp \! \left(\!-\frac{\theta_3}{\lambda_{\mathcal{R}\mathcal{D}}}\!\right)
  \nonumber \\
 &  + \frac{1}{\lambda_{\mathcal{R}\mathcal{D}}}\!\!\int_{\theta_2}^{\theta_3}\!\!\!\exp\left(\!\!-\left(\!\frac{\psi(x)}{\lambda_{\mathcal{S}\mathcal{R}}} + \frac{x}{\lambda_{\mathcal{R}\mathcal{D}}}\!\right)\!\right)\!\mathrm{d}x\Bigg]
\label{eq:sec_pos2}\\
&\!\!\approx (1 - P_{\mathrm{p,out}})\exp\!\!\left(\!\!-\sqrt{\frac{\theta_2}{\lambda^{2}_{\mathcal{R}\mathcal{D}}}}\right)\!,\, (\text{high SNR approximation}),\label{eq:sec_pos_apprx} 
\end{align}
\label{eq:ppos}
\end{subequations}
where
\begin{subequations}
\begin{equation}
\theta_2 = \mathcal{A},
\end{equation}
\begin{align}
\theta_3 &= \left(\frac{\mathcal{B}}{2} + \sqrt{\left(\frac{\mathcal{B}}{2}\right)^{2} + \left(-\frac{\mathcal{A}}{3}\right)^{3}}\right)^{\frac{1}{3}} \nonumber \\
&+ \left(\frac{\mathcal{B}}{2} - \sqrt{\left(\frac{\mathcal{B}}{2}\right)^{2} + \left(-\frac{\mathcal{A}}{3}\right)^{3}}\right)^{\frac{1}{3}},
\label{eq:theta2}
\end{align}
with $\mathcal{A} = \frac{N_0}{\eta \beta P}$ and $\mathcal{B} = \frac{N_0^2}{\eta \beta (1-\beta) P^2}$, and 
\begin{equation}
\psi(x) = \frac{N_0^2}{P(1-\beta)(\eta \beta P x^2 - N_0)} - x.
\end{equation}
\end{subequations}
with
\begin{equation}
\label{eq:binary}
\psi(x)  \left\{
  \begin{array}{l l}
    < 0, & \quad 0 \leq x < \theta_2 \\
    \geq 0, & \quad \theta_2 \leq x \leq \theta_3,\\
  < 0, & \quad \theta_3 < x < \infty. \\
  \end{array} \right.
\end{equation}
$\theta_2$ is the positive root of the equation $g(x) = \eta \beta Px^{2} - N_0 = 0$, while $\theta_3$ is the real root of $\psi(x) = 0$ that is equivalent to a cubic equation given as $x^3 - \mathcal{A}x - \mathcal{B} = 0$.
\end{proposition}

\begin{proof}
See Appendix \ref{app:ppos}.
\end{proof}

\subsubsection{Ergodic Secrecy Rate}
Another important secrecy metric is the ergodic secrecy rate, which is the maximum transmission rate at which the eavesdropper fails to decode the secret information that is being transmitted. We can obtain the ergodic secrecy rate by averaging out the instantaneous secrecy rate $R_{\mathrm{sec}}$ over all possible channel realizations. Therefore, in the case of untrusted relaying, the ergodic secrecy rate, with the inclusion of power outage probability $P_{\mathrm{p,out}}$ given by \eqref{eq:eout}, can be given as
\begin{align}
\bar{R}_{\mathrm{sec}}& = (1 - P_{\mathrm{p,out}}) \mathbb{E}\lbrace R_{\mathrm{sec}} \rbrace  \nonumber \\
&= (1 - P_{\mathrm{p,out}}) \mathbb{E}\left\lbrace\frac{1}{2}\left[\log_2\left(\frac{1 + \gamma_{\mathcal{D}}}{1 + \gamma_{\mathcal{R}}}\right)\right]^{+} \right\rbrace,
\label{eq:esc1}
\end{align}
where $\mathbb{E}\lbrace\cdot\rbrace$ is the expectation operator. Using \eqref{eq:snr_rel} and \eqref{eq:snr_des} in \eqref{eq:esc1}, we can write the analytical expression for $\bar{R}_{\mathrm{sec}}$ as
\begin{align}
\bar{R}_{\mathrm{sec}}& = (1 - P_{\mathrm{p,out}}) \nonumber \\
& \times \int_{x = 0}^{\infty}\int_{y = 0}^{\infty}\! \left[\frac{1}{2}\log_2\left(\frac{1 + \frac{\eta \beta(1-\beta)P xy}{ \eta \beta  N_0 y+ N_0 (1-\beta) +\frac{ N_0^2}{P\left(x +  y\right)} }}{1 + \frac{(1-\beta)Px}{(1-\beta)Py + N_0}}\right)\!\right]^{+} \nonumber \\
& \times f_{|h_{\mathcal{S}\mathcal{R}}|^{2}}(x)f_{|h_{\mathcal{R}\mathcal{D}}|^{2}}(y)\, \mathrm{d}x \, \mathrm{d}y.
\label{eq:exact_erg}
\end{align}
Using high SNR approximation for $\gamma_{\mathcal{D}}$ as given in \eqref{eq:snr_des_appr}, we can write $\bar{R}_{\mathrm{sec}}$ as
\begin{align}
\bar{R}_{\mathrm{sec}} &\approx (1 - P_{\mathrm{p,out}}) \nonumber \\
& \times \int_{x = 0}^{\infty}\int_{y = 0}^{\infty}\left[\frac{1}{2}\log_2\left(\frac{1 + \frac{\eta \beta(1-\beta)P|h_{\mathcal{S}\mathcal{R}}|^{2}|h_{\mathcal{R}\mathcal{D}}|^{2}}{N_0\left(\eta \beta |h_{\mathcal{R}\mathcal{D}}|^{2} + (1-\beta)\right)}}{1 + \frac{(1-\beta)P|h_{\mathcal{S}\mathcal{R}}|^{2}}{(1-\beta)P|h_{\mathcal{R}\mathcal{D}}|^{2} + N_0}}\right)\right]^{+} \nonumber \\
& \times f_{|h_{\mathcal{S}\mathcal{R}}|^{2}}(x)f_{|h_{\mathcal{R}\mathcal{D}}|^{2}}(y)\, \mathrm{d}x \, \mathrm{d}y.
\label{eq:csec}
\end{align}
The expressions in~\eqref{eq:exact_erg} and \eqref{eq:csec} do not admit a closed form and are intractable. Alternatively, we provide a closed-form lower bound on \eqref{eq:csec} as given in the following Proposition. The lower bound on the ergodic secrecy rate ensures the minimum ergodic secrecy rate under all possible channel conditions for a given set of parameters.\footnote{Such guarantee of minimum performance is a useful criterion in the design of a secure communication system.}
\begin{proposition}
The ergodic secrecy rate $\bar{R}_{\mathrm{sec}}$ in \eqref{eq:csec} is lower bounded as
\begin{equation}
\bar{R}_{\mathrm{sec}} \geq (1 - P_{\mathrm{p,out}}) \max\left(\frac{1}{2\ln(2)}(T_1 - T_2), 0\right),
\label{eq:lesc}
\end{equation}
where 
\begin{subequations}
\begin{align}
T_1 &\geq \ln\left(1 + \exp\left(-2\phi - \ln\left(\frac{1}{m_x m_z}\right)\right) \right. \nonumber \\
& \left. + \exp\left(\frac{1}{m_z}\right)+ \mathrm{Ei}\left(-\frac{1}{m_z}\right)\right)
\label{eq:lesc1}
\end{align}
and
\begin{equation}
   \label{eq:T2_fin}
T_2 = \left\{
  \begin{array}{l l}
     1+\frac{1}{m_x}\exp\left(\frac{1}{m_x}\right)\mathrm{Ei}\left(-\frac{1}{m_x}\right), & \quad \frac{m_y}{m_x} = 1\\
    \frac{m_x}{m_x - m_y}\left[\exp\left(\frac{1}{m_y}\right)\mathrm{Ei}\left(-\frac{1}{m_y}\right) \right.  \\
    \left. -  \exp\left(\frac{1}{m_x}\right)\mathrm{Ei}\left(-\frac{1}{m_x}\right)\right], & \quad \frac{m_y}{m_x} \neq 1,\\
    \end{array} \right.
 \end{equation}
 \end{subequations}
\noindent with $m_x = \frac{(1-\beta)P\lambda_{\mathcal{S}\mathcal{R}}}{N_0}$, $m_y  = \frac{(1-\beta)P\lambda_{\mathcal{R}\mathcal{D}}}{N_0}$, $m_z = \frac{\eta \beta \lambda_{\mathcal{R}\mathcal{D}}}{1-\beta}$, $\phi \approx \mathrm{0.577215}$, is the Euler's constant~\cite[9.73]{grad}, and $\mathrm{Ei}(x) = -\int_{-x}^{\infty} \left(\exp(-t)/t\right)\mathrm{d}t$, is the exponential integral~\cite[8.21]{grad}.
\label{prop:esc}
\end{proposition}
\begin{proof}
See Appendix~\ref{app:esc}.
\end{proof}
The lower bound given in \eqref{eq:lesc} is tight in high SNR regime, which is depicted in Fig.~\ref{fig:6} of Section~\ref{sec:results}. Proposition~\ref{prop:esc} shows that the ergodic secrecy rate depends on the power splitting factor $\beta$, energy conversion efficiency factor $\eta$, and mean channel gains of source-to-relay and relay-to-destination links.

\section{Time Switching Policy Based Relaying}
\label{sec:TS}
\begin{figure}
\centering
\includegraphics[scale=0.23]{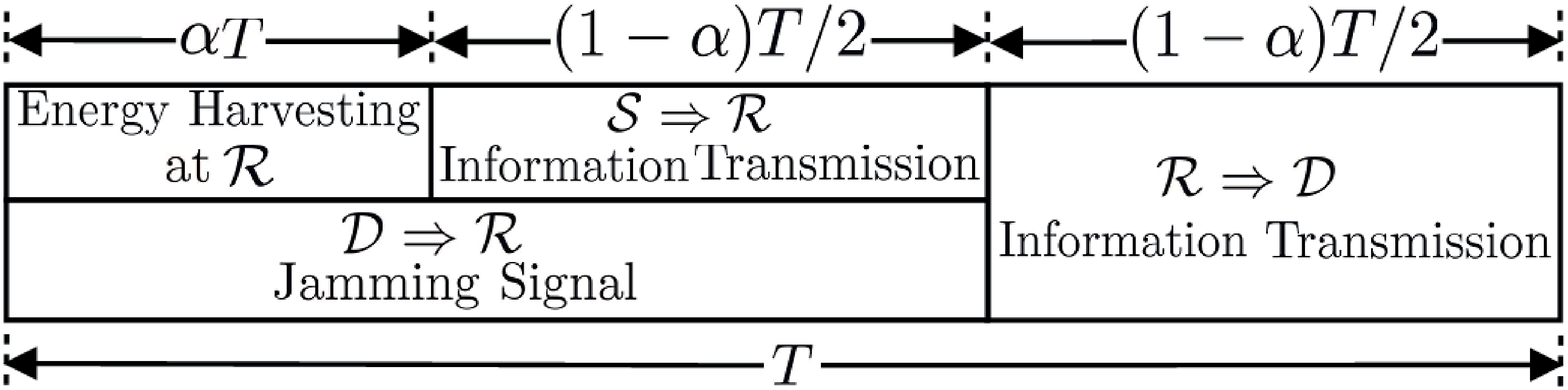}\vspace*{-2mm}
\caption{Time switching policy for the secure communication via an energy harvesting untrusted relay.}\vspace*{-4mm}
\label{fig:TSP}
\end{figure}
Fig.~\ref{fig:TSP} shows the time switching (TS) policy based relaying protocol for the secure communication via untrusted relay. The communication between the source and the destination happens over two hops and in a duration of $T$. The relay harvests energy for $\alpha T$ duration ($0 \leq \alpha \leq 1$) from the received RF signals. The relay spends its harvested energy to forward the received information from the source to the destination. The remaining $(1-\alpha)T$ duration is further split in two sub-slots of equal duration of $\frac{(1-\alpha)T}{2}$. In the first sub-slot, the source transmits the information to relay, which is forwarded to the destination in the second sub-slot after the amplification. The destination sends a jamming signal during the source-to-relay transmission.

\subsection{Energy Harvesting at Relay}
For the aforementioned TS policy, the energy $E_{H}$ harvested during $\alpha T$ duration is given by
\begin{equation}
E_{H} = \eta \alpha T \left(P_{\mathcal{S}} |h_{\mathcal{S}\mathcal{R}}|^2 + P_{\mathcal{D}} |h_{\mathcal{D}\mathcal{R}}|^2\right).
\end{equation} 
The relay uses this harvested energy to forward the source information to the destination with power given by
\begin{equation}
P_{H} = \frac{E_{H}}{(1-\alpha)T/2} = \frac{2 \eta \alpha  \left(P_{\mathcal{S}} |h_{\mathcal{S}\mathcal{R}}|^2 + P_{\mathcal{D}} |h_{\mathcal{D}\mathcal{R}}|^2\right)}{1-\alpha}.
\label{eq:harv_power_TS}
\end{equation}

\subsection{Information Processing and Relaying Protocol}
After the energy harvesting phase, the relay switches to information processing phase, where the received signal is given by
\begin{equation}
y_{\mathcal{R}} = \sqrt{P_{\mathcal{S}}}h_{\mathcal{S}\mathcal{R}}x_{\mathcal{S}} + \sqrt{P_{\mathcal{D}}}h_{\mathcal{D}\mathcal{R}}x_{\mathcal{D}} + n_{\mathcal{R}}.
\label{eq:recTS}
\end{equation}
Note that, unless otherwise stated, all notations in this section have the same meanings as they have in Section~\ref{sec:PS} on the power splitting policy based relaying. Using the received signal $y_{\mathcal{R}}$ given in \eqref{eq:recTS}, the relay may attempt to decode source information. The received SNR at the relay is given by
\begin{equation}
\gamma_{\mathcal{R}} = \frac{P_{\mathcal{S}}|h_{\mathcal{S}\mathcal{R}}|^{2}}{P_{\mathcal{D}}|h_{\mathcal{D}\mathcal{R}}|^{2} + N_0}.
\label{eq:snr_relTS}
\end{equation}
The relay forwards the amplified version of the received signal to the destination, which is given by
\begin{align}
x_{\mathcal{R}} = \xi y_{\mathcal{R}}  =\sqrt{\frac{P_{H}}{P_{\mathcal{S}}|h_{\mathcal{S}\mathcal{R}}|^{2}+P_{\mathcal{D}}|h_{\mathcal{D}\mathcal{R}}|^{2} + N_0}}y_{\mathcal{R}}.
\label{eq:ampl1}
\end{align} 
Then the received signal $y'_{\mathcal{D}}$ at the destination is given by
\begin{align}
y'_{\mathcal{D}} &= h_{\mathcal{R}\mathcal{D}}x_{\mathcal{R}} + n_{\mathcal{D}}.
\label{eq:rec_des11}
\end{align}
After subtracting the term corresponding to the known jamming signal $x_{\mathcal{D}}$, the resultant received signal $y_{\mathcal{D}}$ at the destination becomes
\begin{equation}
y_{\mathcal{D}} = \xi\sqrt{P_{\mathcal{S}}}h_{\mathcal{S}\mathcal{R}}h_{\mathcal{R}\mathcal{D}}x_{\mathcal{S}} + \xi h_{\mathcal{R}\mathcal{D}} n_{\mathcal{R}} + n_{\mathcal{D}}.
\label{eq:rec_des12}
\end{equation} 
Substituting $P_{H}$ from~\eqref{eq:harv_power_TS}~in~\eqref{eq:ampl1}, and then $\xi$ from \eqref{eq:ampl1} in~\eqref{eq:rec_des12}, we can write the received signal $y_{\mathcal{D}}$ as
\begin{align}
y_{\mathcal{D}} &= \frac{\sqrt{2 \eta \alpha P_{\mathcal{S}} \left(P_{\mathcal{S}} |h_{\mathcal{S}\mathcal{R}}|^2 + P_{\mathcal{D}} |h_{\mathcal{D}\mathcal{R}}|^2\right)}h_{\mathcal{S}\mathcal{R}}h_{\mathcal{R}\mathcal{D}}x_{\mathcal{S}}}{\sqrt{(1-\alpha)(P_{\mathcal{S}}|h_{\mathcal{S}\mathcal{R}}|^{2}+P_{\mathcal{D}}|h_{\mathcal{D}\mathcal{R}}|^{2} + N_0)}} \nonumber \\
& + \frac{\sqrt{2 \eta \alpha \left(P_{\mathcal{S}} |h_{\mathcal{S}\mathcal{R}}|^2 + P_{\mathcal{D}} |h_{\mathcal{D}\mathcal{R}}|^2\right)}h_{\mathcal{R}\mathcal{D}}n_{\mathcal{R}}}{\sqrt{(1-\alpha)(P_{\mathcal{S}}|h_{\mathcal{S}\mathcal{R}}|^{2}+P_{\mathcal{D}}|h_{\mathcal{D}\mathcal{R}}|^{2} + N_0)}} + n_{\mathcal{D}}.
\label{eq:rec_des22}
\end{align}
The first term on the right hand side of \eqref{eq:rec_des22} represents the received signal part at the destination, while the last two terms represent the overall noise at the destination. Thus, we can write the received SNR at the destination as
\begin{align}
\gamma_{\mathcal{D}} = \frac{2 \eta \alpha P_{\mathcal{S}}|h_{\mathcal{S}\mathcal{R}}|^{2}|h_{\mathcal{R}\mathcal{D}}|^{2}}{2 \eta \alpha |h_{\mathcal{R}\mathcal{D}}|^2N_0 + N_0 (1-\alpha) +  \frac{N_0^2 (1-\alpha)}{ \left(P_{\mathcal{S}} |h_{\mathcal{S}\mathcal{R}}|^2 + P_{\mathcal{D}} |h_{\mathcal{D}\mathcal{R}}|^2\right)}}.
\label{eq:snr_desTS}
\end{align}
For the rest of the Section~\ref{sec:TS}, we assume $P_{\mathcal{S}} = P_{\mathcal{D}} = P$ for analytical tractability.

\subsection{Secure Communication Via an Untrusted Relay}

For the proposed TS policy, the instantaneous secrecy rate can be given by
\begin{align}
R_{\mathrm{sec}} &= \frac{(1-\alpha)}{2}\bigg[\log_2\left(1 +\gamma_{\mathcal{D}}\right) - \log_2\left(1 +\gamma_{\mathcal{R}}\right)\bigg]^{+} \nonumber \\
&=\frac{(1-\alpha)}{2}\left[\log_2\left(\frac{1 + \gamma_{\mathcal{D}}}{1 + \gamma_{\mathcal{R}}}\right)\right]^{+},
\label{eq:sec_capTS}
\end{align}
where $\gamma_{\mathcal{R}}$ and $\gamma_{\mathcal{D}}$ are given by \eqref{eq:snr_relTS} and \eqref{eq:snr_desTS}, respectively. The factor $(1-\alpha)/2$ denotes the effective time of information transmission between source and destination.

\subsubsection{Secrecy Outage Probability} 
We can express the secrecy outage probability as given in the Proposition~\ref{prop:out_exTS}.
\begin{proposition}
\label{prop:out_exTS}
For TS policy, given the energy harvesting circuitry of the relay is active, the secrecy outage probability is analytically given by~\eqref{eq:out_fin},
where $\delta = 2^{\frac{2R_{\mathrm{th}}}{1-\alpha}}$ with
\begin{equation}
\theta_1 = \frac{(\delta-1) + \sqrt{\left(\delta-1\right)^{2} + 4\delta \frac{P (1-\alpha)}{2 \eta \alpha N_0}}}{2(P/N_0)},
\end{equation} and  
\begin{equation}
\nu(x) = \left(\frac{2 \eta \alpha P x}{N_0 \left(2 \eta \alpha x + (1-\alpha)\right)}-\frac{P\delta}{Px + N_0}\right).
\end{equation} 
\end{proposition}
\begin{proof}
The proof follows the same steps used in Appendix~\ref{app:out_ex} to derive the secrecy outage probability for PS policy in Proposition \ref{prop:out_ex}. Thus, we skip the proof for TS policy for brevity. 
\end{proof}
Note that, for TS policy, the secrecy outage probability under high SNR approximation as given by \eqref{eq:out_fin} is obtained by approximating the exact expression of $\gamma_{\mathcal{D}}$ in \eqref{eq:snr_desTS} as
\begin{equation}
\gamma_{\mathcal{D}} \approx \frac{2 \eta \alpha P|h_{\mathcal{S}\mathcal{R}}|^{2}|h_{\mathcal{R}\mathcal{D}}|^{2}}{N_0\left(2 \eta \alpha |h_{\mathcal{R}\mathcal{D}}|^{2} + (1-\alpha)\right)},
\label{eq:snr_des_apprTS}
\end{equation}
where we have used the channel reciprocity, i.e., $h_{\mathcal{R}\mathcal{D}} = h_{\mathcal{D}\mathcal{R}}$. We have obtained \eqref{eq:snr_des_apprTS} from the exact expression of received SNR at the destination given in \eqref{eq:snr_desTS} by neglecting the term $\frac{N_0^2 (1-\alpha)}{ \left(P_{\mathcal{S}} |h_{\mathcal{S}\mathcal{R}}|^2 + P_{\mathcal{D}} |h_{\mathcal{D}\mathcal{R}}|^2\right)}$ in the denominator of \eqref{eq:snr_desTS} due to negligible value of $N_0^2$ at high SNR. Now, considering the power outage probability, we can finally write the total secrecy outage probability as \eqref{eq:tot_out}. Note that the power outage probability for PS and TS policies is the same. 
\subsubsection{Probability of Positive Secrecy Rate}
The following proposition gives the analytical expression for $P_{\mathrm{pos}}$.
\begin{proposition}

We can write $P_{\mathrm{pos}}$ as \eqref{eq:ppos}, where $\theta_2 = \mathcal{A}$, $\theta_3$ is given by \eqref{eq:theta2}
with $\mathcal{A} = \frac{N_0 (1-\alpha)}{2 \eta \alpha P}$ and $\mathcal{B} = \frac{N_0^2 (1-\alpha)}{2 \eta \alpha P^2}$, and 
\begin{equation*}
\psi(x) = \frac{N_0^2}{P\left(\frac{2 \eta \alpha}{1-\alpha} P x^2 - N_0\right)} - x.
\end{equation*}
\end{proposition}
\begin{proof}

The proof follows the same steps used in Appendix~\ref{app:ppos} for PS policy. We skip the proof for TS policy for brevity.
\end{proof}
\subsubsection{Ergodic Secrecy Rate}
With the inclusion of the power outage probability $P_{\mathrm{p,out}}$ given in \eqref{eq:eout}, the ergodic secrecy rate is calculated by averaging the instantaneous secrecy rate over all possible channel realizations and is given as
\begin{align}
\bar{R}_{\mathrm{sec}}& = (1 - P_{\mathrm{p,out}})\mathbb{E}\lbrace R_{\mathrm{sec}} \rbrace  \nonumber \\
&= (1 - P_{\mathrm{p,out}}\mathbb{E}\left\lbrace\frac{(1-\alpha)}{2}\left[\log_2\left(\frac{1 + \gamma_{\mathcal{D}}}{1 + \gamma_{\mathcal{R}}}\right)\right]^{+} \right\rbrace.
\label{eq:esc1TS}
\end{align}
Using \eqref{eq:snr_relTS} and \eqref{eq:snr_desTS} in \eqref{eq:esc1TS}, we can write the analytical expression for $\bar{R}_{\mathrm{sec}}$ as\vspace*{-3mm}

{{\small
\begin{align}
\!\bar{R}_{\mathrm{sec}} &\!=(1 - P_{\mathrm{p,out}}) \nonumber \\
& \!\! \times \int_{x = 0}^{\infty}\!\int_{y = 0}^{\infty}\! \left[\!\frac{(1-\alpha)}{2}\log_2 \!\! \left(\!\frac{1 + \frac{2 \eta \alpha P xy}{2 \eta \alpha  N_0 y+ N_0 (1-\alpha) +\frac{ N_0^2 (1-\alpha)}{P\left(x +  y\right)} }}{1 + \frac{Px}{Py + N_0}}\!\right)\!\!\right]^{+} \nonumber\\
& \times f_{|h_{\mathcal{S}\mathcal{R}}|^{2}}(x)f_{|h_{\mathcal{R}\mathcal{D}}|^{2}}(y)\, \mathrm{d}x \, \mathrm{d}y.
\label{eq:csecTS_exact}
\end{align}}}\vspace*{-3mm}

\noindent Using high SNR approximation for $\gamma_{\mathcal{D}}$ as given in \eqref{eq:snr_des_apprTS}, we can write $\bar{R}_{\mathrm{sec}}$ as
\begin{align}
\bar{R}_{\mathrm{sec}} &\approx (1 - P_{\mathrm{p,out}}) \nonumber \\
& \times \int_{x = 0}^{\infty}\int_{y = 0}^{\infty} \left[\frac{(1-\alpha)}{2}\log_2\left(\frac{1 + \frac{2 \eta \alpha P xy}{2 \eta \alpha  N_0 y+ N_0 (1-\alpha) }}{1 + \frac{Px}{Py + N_0}}\right)\right]^{+} \nonumber \\
& \times f_{|h_{\mathcal{S}\mathcal{R}}|^{2}}(x)f_{|h_{\mathcal{R}\mathcal{D}}|^{2}}(y)\, \mathrm{d}x \, \mathrm{d}y.
\label{eq:csecTS}
\end{align}
Both~\eqref{eq:csecTS_exact} and \eqref{eq:csecTS} do not admit a closed form. Alternatively, we present a closed-form lower bound on~\eqref{eq:csecTS} as given in the following Proposition.
 \begin{proposition}
We lower bound the ergodic secrecy rate $\bar{R}_{\mathrm{sec}}$ in \eqref{eq:csecTS} by
\begin{equation}
\bar{R}_{\mathrm{sec}} \geq (1 - P_{\mathrm{p,out}})\max\left(\frac{1-\alpha}{2\ln(2)}(T_1 - T_2), 0\right),
\label{eq:lesc_TS}
\end{equation}
where $T_1$ and $T_2$ are given by~\eqref{eq:lesc1} and \eqref{eq:T2_fin}, respectively, with 
$m_x = \frac{P\lambda_{\mathcal{S}\mathcal{R}}}{N_0},$
$m_y  = \frac{P\lambda_{\mathcal{R}\mathcal{D}}}{N_0},$ and 
$m_z = \frac{2 \eta \alpha \lambda_{\mathcal{R}\mathcal{D}}}{1-\alpha}$.
\label{prop:escTS}
\end{proposition}
\begin{proof}
The proof follows the same steps used in Appendix~\ref{app:esc} to derive the lower bound on ergodic secrecy capacity for PS policy in Proposition \ref{prop:esc}. Thus, we skip the proof for TS policy for brevity.
\end{proof}
The lower bound given in \eqref{eq:lesc_TS} is tight in high SNR regime, which is depicted in Fig.~\ref{fig:6} of Section~\ref{sec:results}.

\section{Discussions and Results}
\label{sec:results}
 In this section, we numerically investigate the secrecy performance of source-destination communication via an untrusted wireless energy harvesting relay. For different system parameters like the power splitting ratio, energy harvesting time, transmit SNR, source-relay and relay-destination distances, target secrecy rate, path-loss exponent, and the energy conversion efficiency factor, we discuss how they impact the secrecy outage probability and ergodic secrecy rate under both PS and TS policies.
 
\subsection{System Parameters and Simulation Setup}
Unless otherwise stated, we consider following system parameters. The source power and destination jamming signal power, $P_{S} = P_{D} = P = \mathrm{40}~\mathrm{dBm}$; energy conversion efficiency, $\eta = \mathrm{0.7}$; energy harvesting circuitry activation threshold, $\theta_H = -30~\mathrm{dBm}$~\cite{lu,guo2015}; and noise power, $N_0 = 10^{-4}$. The distances between source and relay and that between relay and destination are $5\mathrm{m}$ each, i.e., $d_{\mathcal{SR}} = d_{\mathcal{RD}} = 5\mathrm{m}$. The mean channel power gains $\lambda_{\mathcal{SR}}$ and $\lambda_{\mathcal{RD}}$ of the exponential random variables $|h_{\mathcal{SR}}|^2$ and $|h_{\mathcal{RD}}|^2$ are $d_{\mathcal{SR}}^{-\rho}$ and $d_{\mathcal{RD}}^{-\rho}$, respectively, where $\rho$ is the path-loss exponent. Unless otherwise stated, $\rho = 2.7$.

\subsection{Effect of power splitting ratio $\beta$ and energy harvesting time $\alpha$}
\subsubsection{Effect of $\beta$} Fig.~\ref{fig:1} shows the effects of the power splitting ratio $\beta$ under PS policy and the energy harvesting time $\alpha$ under TS policy on the secrecy outage probability. For PS policy, with the increase in $\beta$, the secrecy outage probability initially decreases to a minimum value. The value of $\beta$ corresponding to the minimum secrecy outage probability is the optimal value of $\beta$. If we increase $\beta$ further beyond the optimal value, the secrecy outage probability also increases. This is because, as $\beta$ increases, the relay harvests more energy, which in turn, increases the relay's transmit power improving the information reception at the destination. Also, the increased $\beta$ reduces the received signal strength at the relay which degrades the received SNR $\gamma_{\mathcal{R}}$ at the relay. This enhances the secrecy rate of the communication which reduces the secrecy outage probability. But, once $\beta$ crosses the optimal value, the poor signal strength at the relay delivers a negative effect on the secrecy outage probability. Due to the amplification of the poor received signal, the relay forwards a noisy signal to the destination which reduces the received SNR $\gamma_{\mathcal{D}}$ at the destination. The increased harvested energy due to the increased $\beta$, in turn, the higher transmit power of the relay, cannot compensate the loss in $\gamma_{\mathcal{D}}$ because of the reduced signal strength. This pushes the secret source-destination communication into the outage more often, increasing the secrecy outage probability. On the similar line, for Fig.~\ref{fig:4}, we can explain the initial increase of the ergodic secrecy rate with $\beta$ and then its fall after the optimal $\beta$. Figs.~\ref{fig:1} and \ref{fig:4} also show that the simulation results are in excellent agreement with analytical results.

\begin{figure}
\centering
\includegraphics[scale=0.44]{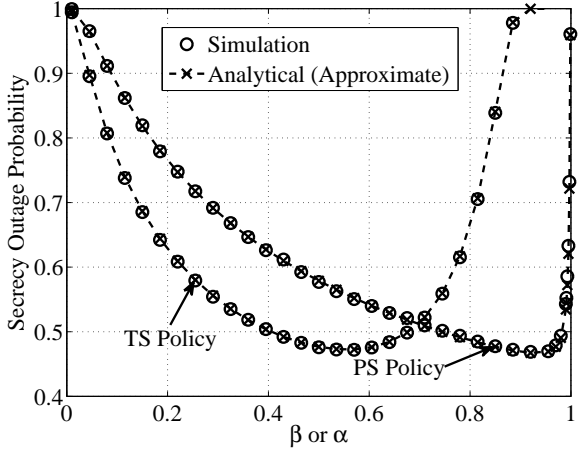}
\caption{Effect of the power splitting ratio $\beta$ and the energy harvesting time $\alpha$ for PS and TS policies, respectively, on the secrecy outage probability, $R_{\mathrm{th}} = \mathrm{0.5}~\mathrm{bits/s/Hz}$.}
\label{fig:1}\vspace*{-3mm}
\end{figure}
\begin{figure}
\centering
\includegraphics[scale=0.44]{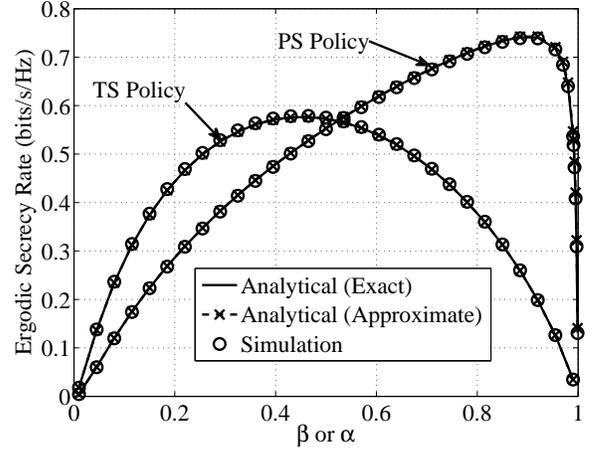}
\caption{Effect of the power splitting ratio $\beta$ and the energy harvesting time $\alpha$ for PS and TS policies, respectively, on the ergodic secrecy rate.}
\label{fig:4}\vspace*{-3mm}
\end{figure}
\subsubsection{Effect of $\alpha$} Fig.~\ref{fig:1} shows that, for TS policy, as the energy harvesting time $\alpha$ increases, the secrecy outage probability reduces initially and reaches the minimum value for the optimal value of $\alpha$. However, the secrecy outage probability begins to increase as $\alpha$ increases beyond its optimal value. This is because, as $\alpha$ increases, the relay spends more time on the energy harvesting, which in turn, increases its transmit power improving the received SNR at the destination. Meanwhile, the increase in $\alpha$ reduces the time available for information processing at both the relay and destination. Now, at the relay, the reduced time for information processing has two opposite effects on the secrecy outage probability. Firstly, it degrades the reception of the signal at the relay and thus deteriorates the eavesdropping channel of the relay improving the secrecy outage probability. On the contrary, since the relay amplifies and forwards the received signal to the destination, the reception at the destination also degrades. Now, when $\alpha$ is less than its optimal value and increasing, the positive effects due to the increased harvested energy at the relay and deterioration of the eavesdropping channel are dominant, and the secrecy outage probability reduces. Once $\alpha$ crosses the optimal value, the effect of the reduced time for information processing becomes dominant, increasing the secrecy outage probability. Similarly, for Fig.~\ref{fig:4}, we can explain the initial increase of the ergodic secrecy rate with $\alpha$ and then its fall after the optimal $\alpha$.

 \begin{figure}
\centering
\includegraphics[scale=0.44]{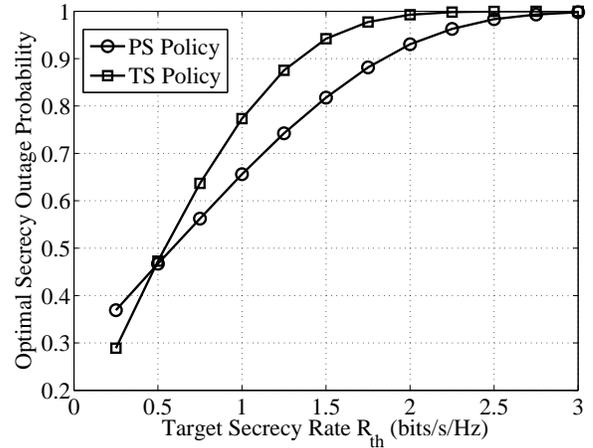}
\caption{Effect of target secrecy rate on the optimal secrecy outage probability for PS and TS policies.}\vspace*{-3mm}
\label{fig:2}
\end{figure}
 
\subsection{Effect of Target Secrecy Rate $R_{\mathrm{th}}$}
Fig.~\ref{fig:2} plots the optimal secrecy outage probability versus the target secrecy rate $R_{\mathrm{th}}$. As the required secrecy rate constraint becomes tighter, the optimal secrecy outage probability increases. This is because, the higher $R_{\mathrm{th}}$ is set, the more it becomes difficult to satisfy, and the likelihood of the secure communication between the source and the destination running into the outage increases. Fig.~\ref{fig:2} also shows that TS policy achieves lower secrecy outage probability at low $R_{\mathrm{th}}$ (till $\mathrm{0.5~\mathrm{bits/s/Hz}}$) than that of PS policy. On the contrary, at higher secrecy rate constraint, PS policy outperforms TS policy.

\subsection{Effect of Transmit Signal-to-Noise Ratio (SNR)}

Fig.~\ref{fig:3} illustrates the effect of the transmit SNR, i.e., $P/N_0$, on the optimal secrecy outage probability for both PS and TS policies. For a fixed noise power $N_0$, the variation in transmit SNR is equivalent to the variation of source's and destination's power $P$. The increase in transmit SNR has its constructive as well as destructive effects on the secure communication. The increase in transmit SNR increases the signal strengths of both information signal from the source and jamming signal from the destination. From the expressions of received SNR $\gamma_{\mathcal{R}}$ at the relay given by \eqref{eq:snr_rel} and \eqref{eq:snr_relTS} for PS and TS policies, respectively, we can note that $\gamma_{\mathcal{R}}$ increases with the increase in transmit SNR. This increases the chances of the untrusted relay decoding the information, which leads to the increase in the secrecy outage probability. On the other hand, the increase in transmit SNR increases the energy harvested by the relay due to higher received powers from information and jamming signals. This causes an increase in the relay's transmit power, which improves SNR at the destination. Also, when relay amplifies and forwards its received signal to the destination, the signal strength is further improved due to the increased signal strength at the relay as a result of the increased transmit SNR. As Fig.~\ref{fig:3} shows, the increase in transmit SNR has an overall positive impact on the secrecy performance of the system.

\begin{figure}
\centering
\includegraphics[scale=0.44]{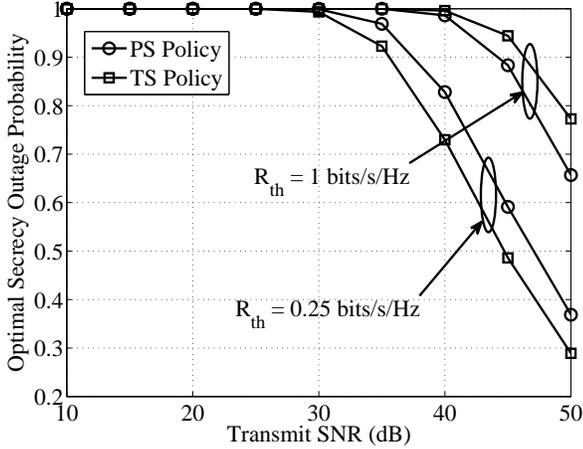}
\caption{Optimal secrecy outage probability versus transmit SNR ($P/N_0$) for PS and TS policies, $N_0 = -10~\mathrm{dBm}$.}\vspace*{-4mm}
\label{fig:3}
\end{figure}

 \begin{figure}
\centering
\includegraphics[scale=0.44]{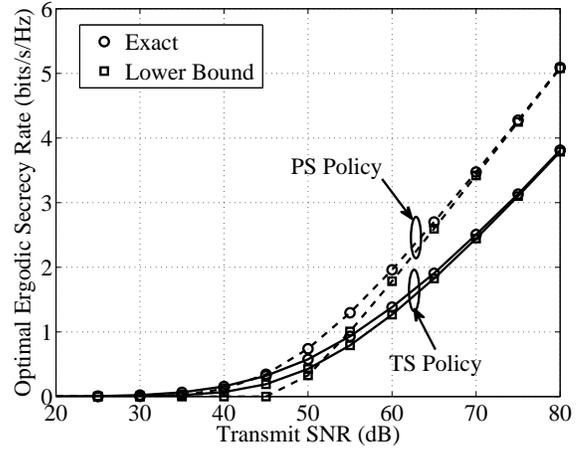}
\caption{Optimal ergodic secrecy rate versus transmit SNR ($P/N_0$) for PS and TS policies, $N_0 = -10~\mathrm{dBm}$.}\vspace*{-4mm}
\label{fig:6}
\end{figure}

Similarly, Fig.~\ref{fig:6} shows that the optimal ergodic secrecy rate improves with the increase in transmit SNR. One interesting observation is that, at lower transmit SNR values, TS policy achieves better optimal ergodic secrecy rate than that of PS policy. On the other hand, at higher transmit SNR, PS policy attains higher ergodic secrecy rate compared to TS policy. From Fig.~\ref{fig:6}, we can note that, with the increase in transmit SNR, the performance with the closed-form lower bound on the ergodic secrecy rate approaches the performance with the exact analytical expression. Thus, the closed-form lower bound is tight at high transmit SNR for both PS and TS policies.

 \begin{figure}
\centering
\includegraphics[scale=0.44]{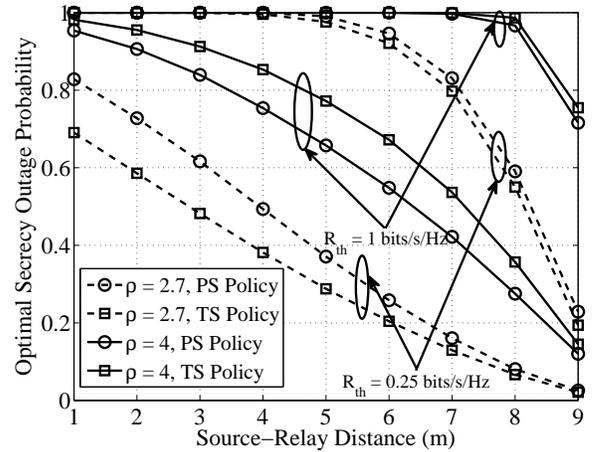}
\caption{Effect of relay placement on the optimal secrecy outage probability for PS and TS policies with different path-loss exponents $\rho = 2.7, 4$.}\vspace*{-3mm}
\label{fig:7}
\end{figure}

\subsection{Effect of Relay Placement}
Fig.~\ref{fig:7} depicts the effect of the relay placement on the optimal secrecy outage probability for different target secrecy rates and path-loss exponents $\rho$ under both PS and TS policies. We vary the source-relay distance $d_{\mathcal{SR}}$, while the relay-destination distance $d_{\mathcal{RD}}$ is $\mathrm{10} - d_{\mathcal{SR}}$. The values of path-loss exponent $\rho$ considered are $\rho = 2.7$ and $4$. Before discussing Fig.~\ref{fig:7}, it is important to understand how $d_{\mathcal{SR}}$ affects the secrecy performance in both constructive and destructive ways. Under both PS and TS policies, as $d_{\mathcal{SR}}$ increases, the received information signal strength at the relay decreases due to the higher path-loss $d_{\mathcal{SR}}^{-\rho}$. This discourages the eavesdropping intention of the untrusted relay, improving the secrecy performance. Also, as $d_{\mathcal{SR}}$ increases, the relay-destination distance $d_{\mathcal{RD}}$ reduces, which makes the received jamming signal at the relay stronger. This further enhances the secrecy performance. The decrease in $d_{\mathcal{RD}}$ brings the relay closer to the destination due to which the lesser amount of harvested energy is sufficient to perform the reliable communication between relay and destination because of the reduced path-loss $d_{\mathcal{RD}}^{-\rho}$. This saving in the energy is important as, the energy harvested by the relay decreases with the increase in $d_{\mathcal{SR}}$. Another negative effect of the increased $d_{\mathcal{SR}}$ on the secrecy performance is that, due to the amplify-and-forward nature of the relay, as the received signal strength at the relay reduces with the increase in $d_{\mathcal{SR}}$, the information signal strength at the destination also deteriorates. This reduces the secrecy rate and thus increases the secrecy outage probability.

  \begin{figure}
\centering
\includegraphics[scale=0.44]{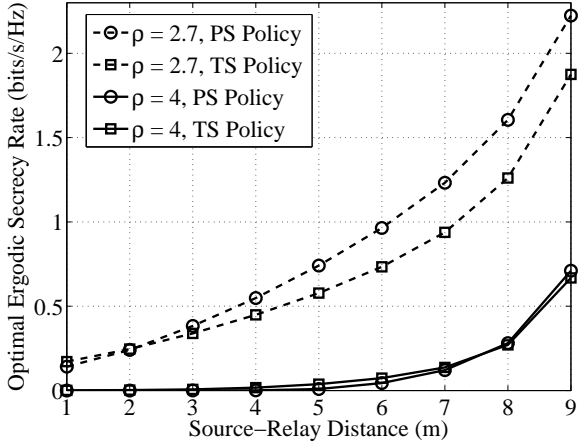}
\caption{Effect of relay placement on the optimal ergodic secrecy rate for PS and TS policies with different path-loss exponents $\rho = 2.7, 4$.}
\label{fig:8}\vspace*{-3mm}
\end{figure}

Fig.~\ref{fig:7} shows that the constructive effects of the increase in $d_{\mathcal{SR}}$ overtake its destructive effects irrespective of the secrecy rate threshold $R_{\mathrm{th}}$ under both PS and TS policies and the optimal secrecy outage probability decreases monotonically with the increase in $d_{\mathcal{SR}}$. Thus, the optimum placement of the relay is closer to the destination.
Note that, in the case of wireless energy harvesting communication via a relay without secrecy constraints, the optimum relay placement is close to the source~\cite{nasir}. But, as shown in Figs.~\ref{fig:7} and \ref{fig:8}, to have secure communication, the relay placement close to the source is not preferred.

Fig.~\ref{fig:8} shows that, for the optimal ergodic secrecy rate, the relay placement has similar effects on the secrecy performance as that on the optimal secrecy outage probability. One interesting observation is that, with the variation in $d_{\mathcal{SR}}$, there exists a crossover point between PS and TS policies, and the location of the crossover point depends on the path-loss exponent. For example, for the path-loss exponent $\rho = 2.7$, TS policy achieves higher optimal ergodic secrecy rate than that of PS policy below $d_{\mathcal{SR}} = 2\mathrm{m}$, i.e., the crossover occurs at $d_{\mathcal{SR}} = 2\mathrm{m}$; while for $\rho = 4$, TS policy achieves higher optimal ergodic secrecy rate than that of PS policy below $d_{\mathcal{SR}} = 8\mathrm{m}$, i.e., the crossover occurs at $d_{\mathcal{SR}} = 8\mathrm{m}$. This is because, at a given path-loss exponent, below the crossover point, the loss in information processing time due to the energy harvesting time in TS policy is lesser than the loss incurred in the relay's transmit power due to power splitting in PS policy. As the distance between relay and destination decreases (with the increase in $d_{\mathcal{SR}}$), the relay may transmit with lower power due to lower path-loss. This subsides the loss incurred in power splitting in PS policy compared to the loss in time for TS policy, and PS policy outperforms TS policy at higher $d_{\mathcal{SR}}$. The increase in path-loss exponent delays the arrival of the crossover point, because, for higher path-loss exponent, the distance between relay and destination should be lower than that in the case of lower path-loss exponent to subside the loss incurred in power splitting. This effect of path-loss exponent on the optimal ergodic secrecy rate can also be seen in Fig.~\ref{fig:new} for different source-relay distances. In addition to the effect of the path-loss exponent on the crossover point, Fig.~\ref{fig:new} shows that the increase in path-loss exponent is detrimental for the secure communication.
\begin{figure}
\centering
\includegraphics[scale=0.44]{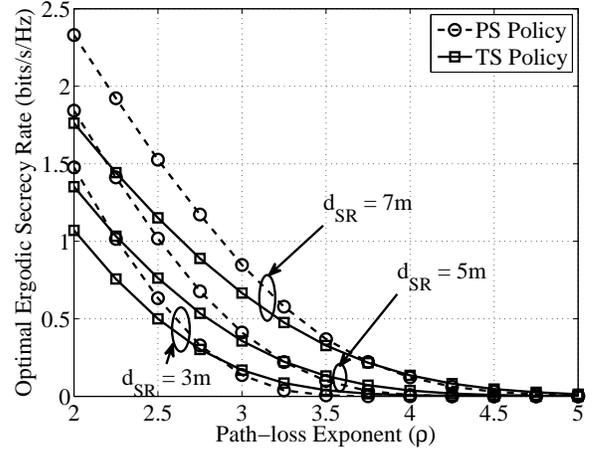}
\caption{Effect of path-loss exponent on the optimal ergodic secrecy rate for PS and TS policies with different source-relay distances $d_{\mathcal{SR}} = 3\mathrm{m}, 5\mathrm{m}, 7\mathrm{m}$.}
\label{fig:new}\vspace*{-3mm}
\end{figure}

\subsection{Effect of Energy Conversion Efficiency Factor $\eta$}
The energy conversion efficiency factor $\eta$ determines what fraction of the received power the relay can actually harvest. Thus, higher $\eta$ allows relay to harvest more energy, which in turn, boosts relay's transmit power. This results in the enhanced received SNR at the destination, reducing the secrecy outage probability and improving the ergodic secrecy rate, as shown in Figs.~\ref{fig:9} and \ref{fig:10}, respectively. At lower $\eta$, TS policy achieves better optimal ergodic secrecy rate than that of PS policy and the trend reverses at higher $\eta$.

\begin{figure}
\centering
    \subfigure[]{\label{fig:9}\includegraphics[scale=0.39]{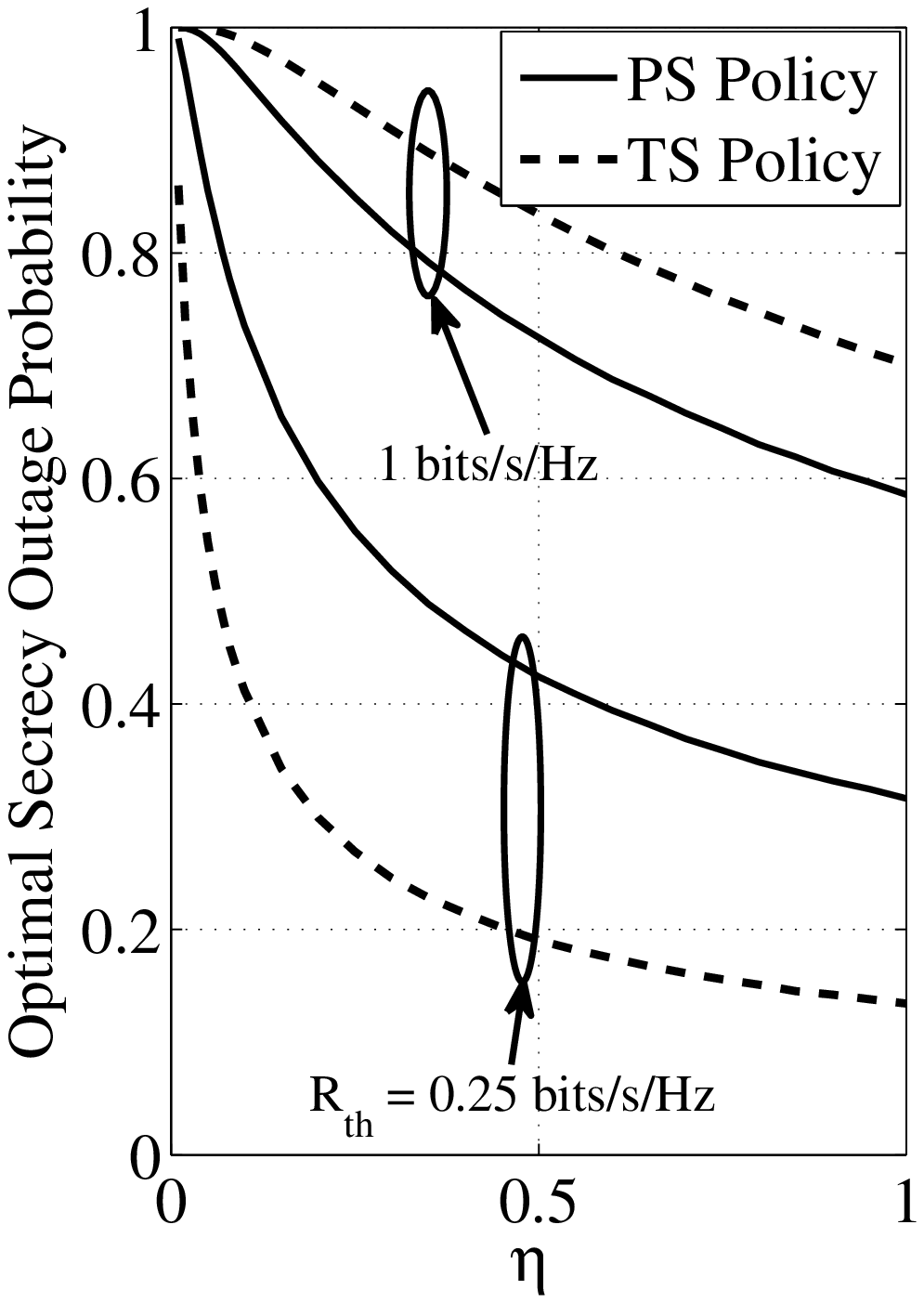}}
    \subfigure[]{\label{fig:10}\includegraphics[scale=0.39]{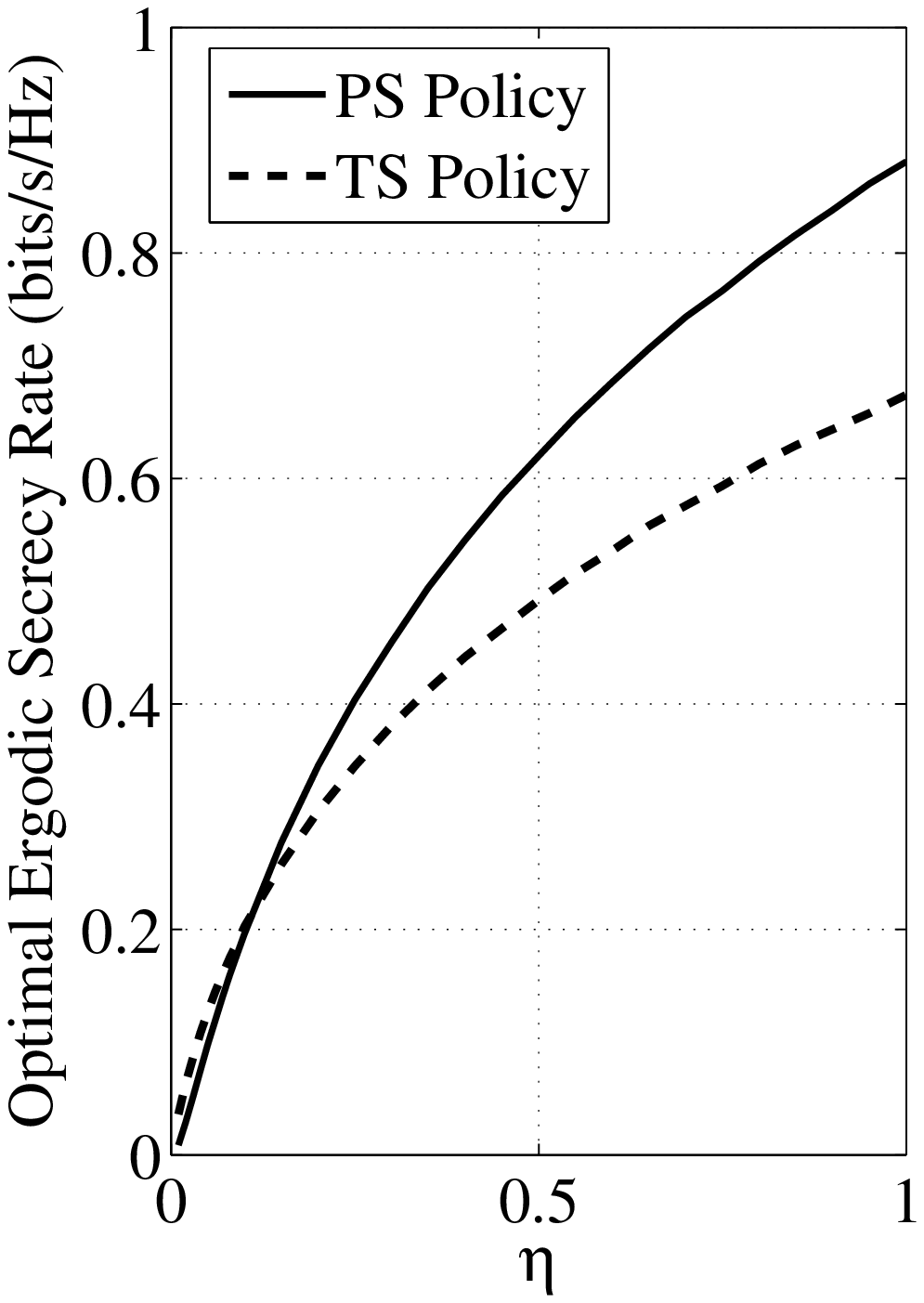}} 
        \caption{Effect of the energy conversion efficiency factor $\eta$ (a) on the optimal secrecy outage probability, (b) on the optimal ergodic secrecy rate.}\vspace*{-3mm}
      \label{fig:eta}
    \end{figure}

\section{Concluding Remarks}
\label{sec:conc}
We have investigated the secrecy performance of the source-destination communication via an energy harvesting amplify-and-forward untrusted relay. The energy-starved relay harvests energy from the received radio-frequency signals. In this case, besides keeping the information confidential from the untrusted relay, the destination-assisted jamming signal supplies energy to relay. This energy augments the energy harvested from the received information signal. The PS and TS policies at the relay enable it to harvest energy and process the received information. For this proposed scenario, we have derived analytical expressions for two secrecy metrics, viz., the secrecy outage probability and the ergodic secrecy rate.

The numerical study of the aforementioned secrecy metrics against different system parameters provides useful design insights. For instance, the variation of power splitting ratio in PS policy and energy harvesting time in TS policy affect the secrecy performance in both constructive and destructive ways. Thus, the optimal power splitting ratio and the optimal energy harvesting time exist, that maximize the secrecy performance in terms of both secrecy metrics. The optimal values of secrecy metrics depend on the system parameters. For example, the higher target secrecy rate we set, the more it becomes difficult to achieve, increasing the optimal secrecy outage probability. Also, at higher target secrecy rate, PS policy outperforms TS policy by achieving the lower optimal secrecy outage probability. Though the increase in transmit SNR increases the possibility of relay decoding the confidential information, the resulting higher harvested energy and the jamming power dominate the negative effect. Thus, the increase in transmit SNR is beneficial to the secure communication. We also observe that, for high transmit SNR, PS policy achieves better ergodic secrecy rate than that of TS policy. The relay location is important in the secure communication. In general, having relay located away from the source is beneficial to keep the information confidential from the relay. This is in contrast with the case of trusted energy harvesting relay, where the relay is preferred to be placed closer to the source. Finally, higher energy conversion efficiency factor increases the harvested energy by the relay, which in turn, improves secrecy performance. In particular, at higher energy conversion efficiency factor, PS policy achieves better optimal ergodic secrecy rate than that of TS policy.

\appendices

\section{Derivation of \eqref{eq:out_fin}}
\label{app:out_ex}
At high SNR, using the channel reciprocity between relay and destination and substituting  $\gamma_{\mathcal{R}}$ from \eqref{eq:snr_rel} and $\gamma_{\mathcal{D}}$ from \eqref{eq:snr_des_appr} in \eqref{eq:sec_cap}, and then using \eqref{eq:sec_cap} and \eqref{eq:outt}, we can write the secrecy outage probability for PS policy as
\begin{align}
P_{\mathrm{out}} &= \mathbb{P}\left(\frac{1 +  \frac{\eta \beta(1-\beta)P|h_{\mathcal{S}\mathcal{R}}|^{2}|h_{\mathcal{R}\mathcal{D}}|^{2}}{N_0\left(\eta \beta |h_{\mathcal{R}\mathcal{D}}|^{2} + (1-\beta)\right)}}{1+ \frac{(1-\beta)P|h_{\mathcal{S}\mathcal{R}}|^{2}}{(1-\beta)P|h_{\mathcal{R}\mathcal{D}}|^{2} + N_0}} < \delta\right), \nonumber \\
&= \mathbb{P}\left(\nu(X)|h_{\mathcal{S}\mathcal{R}}|^{2} < \delta-1\right)\bigg{|_{X = |h_{\mathcal{R}\mathcal{D}}|^{2}}}, 
\label{eq:out22}
\end{align}
where 
\begin{equation}
\nu(x) = (1-\beta)\left(\frac{\eta \beta P x}{N_0 \left(\eta \beta x + (1-\beta)\right)}-\frac{P\delta}{P(1-\beta)x + N_0}\right).
\label{eq:nu}
\end{equation}
Based on the sign of $\nu(X)$, we split~\eqref{eq:out22} as
\begin{align}
P_{\mathrm{out}} &= \mathbb{P}\left(|h_{\mathcal{S}\mathcal{R}}|^{2} < \frac{\delta-1}{\nu(X)}\bigg|\nu(X) \geq 0\right)\mathbb{P}\left(\nu(X) \geq 0\right) \nonumber \\
& + \underset{=\,1}{\underbrace{\mathbb{P}\left(|h_{\mathcal{S}\mathcal{R}}|^{2} \geq \frac{\delta-1}{\nu(X)}\bigg|\nu(X) < 0\right)}}\mathbb{P}\left(\nu(X) < 0\right).
\label{eq:gg1}
\end{align}
In~\eqref{eq:gg1}, $\mathbb{P}\left(|h_{\mathcal{S}\mathcal{R}}|^{2} \geq \frac{\delta-1}{\nu(X)}\bigg|\nu(X) < 0\right) = 1$, because $|h_{\mathcal{SR}}|^2$ being an exponential random variable always takes non-negative values. Also, we have
\begin{equation}
\label{eq:binary1}
\nu(x)  \left\{
  \begin{array}{l l}
    \geq 0, & \quad \text{if}\,\,\,\theta_1 \leq x < \infty \\
    < 0, & \quad \text{if}\,\,\, 0 \leq x <\theta_1,\\
  \end{array} \right.
\end{equation}
where 
\begin{equation}
\theta_1 = \frac{\frac{\delta-1}{1-\beta} + \sqrt{\left(\frac{\delta-1}{1-\beta}\right)^{2} + 4 \delta \frac{P}{\eta \beta N_0}}}{2(P/N_0)}.
\end{equation}
Note that $\theta_{1}$ is the positive root of the equation $\nu(x) = 0$. Using~\eqref{eq:binary1}, we can write \eqref{eq:gg1} as 
\begin{align}
P_{\mathrm{out}} &= \int_{\theta_1}^{\infty}\!\! \left(\! 1- \exp\left(\!-\frac{\delta - 1}{\nu(x)\lambda_{\mathcal{S}\mathcal{R}}}\! \right)\!\!\right) f_{X}(x) \, \mathrm{d}x  + \int_{0}^{\theta_1}\!\!\! f_{X}(x)\, \mathrm{d}x, \nonumber\\
&= \underset{=\,1}{\underbrace{\int_{0}^{\theta_1} f_{X}(x)\, \mathrm{d}x + \int_{\theta_1}^{\infty} f_{X}(x)\, \mathrm{d}x}} \nonumber \\
&- \int_{\theta_1}^{\infty} \left(\exp\left(-\frac{\delta - 1}{\nu(x)\lambda_{\mathcal{S}\mathcal{R}}}\right)\right) f_{X}(x) \, \mathrm{d}x.
\label{eq:out24}
\end{align}
Substituting $f_X(x) = \frac{1}{\lambda_{\mathcal{RD}}}\exp\left(-\frac{x}{\lambda_{\mathcal{RD}}}\right)$ in the third integral of \eqref{eq:out24}, we reach the required expression of $P_{\mathrm{out}}$ as in \eqref{eq:out_fin}.

\section{Proof of Proposition~\ref{prop:eout}}
\label{app:eout}
We can write the power outage probability as
\begin{align}
P_{\mathrm{p,out}}& = \mathbb{P}\left(P_R < \theta_H\right) \nonumber \\
&= \mathbb{P}\left(P(|h_{\mathcal{SR}}|^{2} + |h_{\mathcal{RD}}|^{2}) < \theta_H\right) \nonumber \\
&=\mathbb{P}\left((|h_{\mathcal{SR}}|^{2} + |h_{\mathcal{RD}}|^{2}) < \frac{\theta_H}{P}\right).
\label{eq:eout1}
\end{align}
Let $Z = \left(|h_{\mathcal{SR}}|^2 + |h_{\mathcal{RD}}|^2\right)$. Since $|h_{\mathcal{SR}}|^2$ and $|h_{\mathcal{RD}}|^2$ are exponentially distributed random variables with means $\lambda_{\mathcal{SR}}$ and $\lambda_{\mathcal{RD}}$, we can write the probability density function of $Z$ as~\cite{papoulis}
\begin{equation}
\label{eq:binaryzz1}
f_{Z}(z) =  \left\{
  \begin{array}{l l}
    \frac{\exp\left(-\frac{z}{\lambda_{\mathcal{SR}}}\right)}{\lambda_{\mathcal{SR}}-\lambda_{\mathcal{RD}}}+  \frac{\exp\left(-\frac{z}{\lambda_{\mathcal{RD}}}\right)}{\lambda_{\mathcal{RD}}-\lambda_{\mathcal{SR}}}, & \quad \mathrm{if}\,\, \lambda_{\mathcal{SR}} \neq  \lambda_{\mathcal{RD}} \\
    \left(\frac{1}{\lambda_{\mathcal{SR}}}\right)^{2}z \exp\left(-\frac{z}{\lambda_{\mathcal{SR}}}\right), & \quad \mathrm{if}\,\, \lambda_{\mathcal{SR}} = \lambda_{\mathcal{RD}}.\\
  \end{array} \right.
\end{equation}
Note that $Z$ can take only non-negative values as it is the sum of two exponential random variables. Using \eqref{eq:binaryzz1} in \eqref{eq:eout1}, we can write
\begin{align}
P_{\mathrm{p,out}}& = \mathbb{P}\left(Z < \frac{\theta_H}{P}\right) \nonumber \\
& = \int_{0}^{\frac{\theta_H}{P}} f_Z(z)\mathrm{d}z.
\label{eq:eout2}
\end{align}
Evaluating the integral in \eqref{eq:eout2}, we get the required expression for the power outage probability as in \eqref{eq:eout}.

\section{Proof of Proposition \ref{prop:ppos}}
\label{app:ppos}
\subsection{Proof of \eqref{eq:sec_pos2}}
We can write the probability of achieving the positive secrecy capacity as
\begin{align}
P_{\mathrm{pos}} &= (1 - P_{\mathrm{p,out}})\mathbb{P}\left(R_{\mathrm{sec}} > 0\right) \nonumber \\
&= (1 - P_{\mathrm{p,out}})\mathbb{P}\left(\frac{1}{2}\log_2\left[\frac{\left(1 +\gamma_{\mathcal{D}}\right)}{\left(1 +\gamma_{\mathcal{R}}\right)}  \right]^{+} > 0\right) \nonumber \\
&= (1 - P_{\mathrm{p,out}})\mathbb{P}\left(\gamma_{\mathcal{D}} > \gamma_{\mathcal{R}}\right).
\label{eq:pos11}
\end{align}
Substituting $\gamma_{\mathcal{R}}$ from \eqref{eq:snr_rel} and $\gamma_{\mathcal{D}}$ from \eqref{eq:snr_des} in \eqref{eq:pos11} , we obtain
\begin{align}
\mathbb{P}\left(\gamma_{\mathcal{D}} > \gamma_{\mathcal{R}}\right)& = \mathbb{P}\big[\big(\big(|h_{\mathcal{S}\mathcal{R}}|^{2} + |h_{\mathcal{R}\mathcal{D}}|^{2}\big)  \nonumber \\
&\times P(1-\beta)\big(\eta \beta P|h_{\mathcal{R}\mathcal{D}}|^{4} - N_0\big)\big) > N_0^2\big]. 
\label{eq:sec_pos_exact}
\end{align}
Then we can write
\begin{align}
\mathbb{P}\left(\gamma_{\mathcal{D}} > \gamma_{\mathcal{R}}\right)& = \int_{0}^{\theta_2} F_{|h_{\mathcal{S}\mathcal{R}}|^{2}}(\psi(x))f_{|h_{\mathcal{R}\mathcal{D}}|^{2}}(x)\,\mathrm{d}x \nonumber \\
 &+ \int_{\theta_2}^{\theta_3} \big[1- F_{|h_{\mathcal{S}\mathcal{R}}|^{2}}(\psi(x))\big]f_{|h_{\mathcal{R}\mathcal{D}}|^{2}}(x)\,\mathrm{d}x \nonumber \\
 &+ \int_{\theta_3}^{\infty} \big[1- F_{|h_{\mathcal{S}\mathcal{R}}|^{2}}(\psi(x))\big]f_{|h_{\mathcal{R}\mathcal{D}}|^{2}}(x)\,\mathrm{d}x \nonumber \\
 & =  \frac{1}{\lambda_{\mathcal{R}\mathcal{D}}}\int_{\theta_2}^{\theta_3}\exp\left(-\left(\frac{\psi(x)}{\lambda_{\mathcal{S}\mathcal{R}}} + \frac{x}{\lambda_{\mathcal{R}\mathcal{D}}}\right)\right)\,\mathrm{d}x \nonumber \\
 & + \exp\left(-\frac{\theta_3}{\lambda_{\mathcal{R}\mathcal{D}}}\right),
 \label{eq:ppos12}
\end{align}
where 
\begin{equation}
\psi(x) = \frac{N_0^2}{P(1-\beta)(\eta \beta P x^2 - N_0)} - x
\end{equation}
with
\begin{equation}
\label{eq:binary}
\psi(x)  \left\{
  \begin{array}{l l}
    < 0, & \quad 0 \leq x < \theta_2, \\
    \geq 0, & \quad \theta_2 \leq x \leq \theta_3,\\
  < 0, & \quad \theta_3 < x < \infty. \\
  \end{array} \right.
\end{equation}
$\theta_2$ is the positive root of the equation $g(x) = \eta \beta Px^{2} - N_0 = 0$, and is given as
\begin{equation}
\theta_2 = \sqrt{\frac{N_0}{\eta \beta P}},
\end{equation}
while $\theta_3$ is the real root of $\psi(x) = 0$ which is a cubic equation given as
\begin{equation}
x^{3} - \mathcal{A}x - \mathcal{B} = 0,
\label{eq:cubic}
\end{equation}
where $\mathcal{A} = \frac{N_0}{\eta \beta P}$ and $\mathcal{B} = \frac{N_0^2}{\eta \beta (1-\beta) P^2}$. We obtain the solution to \eqref{eq:cubic} using Cardano's formula~\cite{nathan}, which allows us to find the real root of \eqref{eq:cubic}. The solution is given as
\begin{align}
\theta_3 &= \left(\frac{\mathcal{B}}{2} + \sqrt{\left(\frac{\mathcal{B}}{2}\right)^{2} + \left(-\frac{\mathcal{A}}{3}\right)^{3}}\right)^{\frac{1}{3}} \nonumber \\
& + \left(\frac{\mathcal{B}}{2} - \sqrt{\left(\frac{\mathcal{B}}{2}\right)^{2} + \left(-\frac{\mathcal{A}}{3}\right)^{3}}\right)^{\frac{1}{3}}.
\end{align}
Substituting \eqref{eq:ppos12} in \eqref{eq:pos11}, we get the exact expression of the probability of positive secrecy rate given in \eqref{eq:sec_pos2}.

\subsection{Proof of \eqref{eq:sec_pos_apprx}}
Under high SNR approximation of $\gamma_{\mathcal{D}}$ given in \eqref{eq:snr_des_appr}, using \eqref{eq:pos11}, we can write the probability of positive secrecy rate as
\begin{align}
 P_{\mathrm{pos}} &= (1 - P_{\mathrm{p,out}}) \nonumber \\
 &\hspace*{-7mm}\times \mathbb{P}\left(\frac{\eta \beta(1-\beta)P|h_{\mathcal{S}\mathcal{R}}|^{2}|h_{\mathcal{R}\mathcal{D}}|^{2}}{N_0\left(\eta \beta |h_{\mathcal{R}\mathcal{D}}|^{2} + (1-\beta)\right)} > \frac{(1-\beta)P|h_{\mathcal{S}\mathcal{R}}|^{2}}{(1-\beta)P|h_{\mathcal{R}\mathcal{D}}|^{2} + N_0}\right),
 \label{eq:ppos_appr1}
\end{align}
where we have used $\gamma_{\mathcal{R}}$ from \eqref{eq:snr_rel} with $P_{\mathcal{S}}= P_{\mathcal{D}} = P$ and $h_{\mathcal{D}\mathcal{R}} = h_{\mathcal{R}\mathcal{D}}$ (channel reciprocity between relay and destination). Simplifying \eqref{eq:ppos_appr1}, we obtain
\begin{align}
 P_{\mathrm{pos}} &= (1 - P_{\mathrm{p,out}})\mathbb{P}\left(|h_{\mathcal{R}\mathcal{D}}|^{2} > \sqrt{\frac{N_0}{\eta \beta P}}\right) \nonumber \\
 &= (1 - P_{\mathrm{p,out}})\exp\left(-\sqrt{\frac{\theta_2}{\lambda^{2}_{\mathcal{R}\mathcal{D}}}}\right),
\end{align}
where $\theta_2 = \frac{N_0}{\eta \beta P}$.

\section{Proof of Proposition \ref{prop:esc}}
\label{app:esc}
For PS policy, we can write the ergodic secrecy rate as
\begin{align}
\bar{R}_{\mathrm{sec}} &=  (1-P_{\mathrm{p,out}}) \mathbb{E}\left\lbrace \frac{1}{2} \left[\log_2\left(\frac{1 + \gamma_{\mathcal{D}}}{1 + \gamma_{\mathcal{R}}}\right)\right]^{+} \right\rbrace \\
& \stackrel{\mathrm{(a)}}{\geq} (1-P_{\mathrm{p,out}}) \left[\mathbb{E}\left \lbrace \frac{1}{2} \log_2\left(\frac{1 + \gamma_{\mathcal{D}}}{1 + \gamma_{\mathcal{R}}}\right)   \right \rbrace\right]^{+} \nonumber \\
& \stackrel{\mathrm{(b)}}{=}(1-P_{\mathrm{p,out}}) \max\Bigg(\frac{1}{2 \ln(2)}\Bigg[\underbrace{\mathbb{E}\left \lbrace \ln\left(1 + \frac{XZ}{Z + 1}\right) \right \rbrace}_{T_1} \nonumber \\
&- \underbrace{\mathbb{E}\left \lbrace \ln\left(1 + \frac{X}{Y + 1}\right) \right \rbrace}_{T_2} \Bigg], 0\Bigg),
\label{eq:esc3}
\end{align}
where $X = \frac{(1-\beta)P|h_{\mathcal{S}\mathcal{R}}|^{2}}{N_0}$, $Y = \frac{(1-\beta)P|h_{\mathcal{R}\mathcal{D}}|^{2}}{N_0}$, and $Z = \frac{\eta \beta |h_{\mathcal{R}\mathcal{D}}|^{2}}{1-\beta}$ are the exponential random variables with means $m_x = \frac{(1-\beta)P\lambda_{\mathcal{S}\mathcal{R}}}{N_0}$, $m_y = \frac{(1-\beta)P \lambda_{\mathcal{R}\mathcal{D}}}{N_0}$, and $m_z = \frac{\eta \beta \lambda_{\mathcal{R}\mathcal{D}}}{1-\beta}$, respectively. The inequality ($\mathrm{a}$) is obtained by using the fact $\mathbb{E}\left \lbrace \max(U, V) \right \rbrace \geq \max \left(\mathbb{E}\left \lbrace U \right \rbrace, \mathbb{E}\left \lbrace V \right \rbrace\right)$. Also, to obtain equality ($\mathrm{b}$), we have used $\gamma_{\mathcal{R}}$ from \eqref{eq:snr_rel} and $\gamma_{\mathcal{D}}$ from \eqref{eq:snr_des_appr}. We can further lower bound $T_1$ as 
\begin{align}
T_1 &= \mathbb{E}\left \lbrace \ln\left(1 + \frac{XZ}{Z + 1}\right) \right \rbrace \nonumber \\
& = \mathbb{E}\left \lbrace \ln\left(1 + \exp\left(\ln\left(\frac{XZ}{Z + 1}\right)\right)\right) \right \rbrace \nonumber \\
& \stackrel{\mathrm{(c)}}{\geq}\ln\left(1 + \exp\left(\mathbb{E}\left \lbrace \ln\left( \frac{XZ}{Z + 1}\right)\right \rbrace \right)\right) \nonumber \\
& = \ln\left(1 + \exp\left(\underbrace{\mathbb{E}\left \lbrace \ln\left( XZ\right)\right \rbrace}_{\mathcal{J}_1} - \underbrace{\mathbb{E}\left \lbrace \ln\left( Z+1\right)\right \rbrace}_{\mathcal{J}_2} \right)\right),
\label{eq:T1}
\end{align}
where we have used the convexity of $\ln(1+t\exp(x))$ for $t > 0$ and Jensen's inequality to obtain inequality ($\mathrm{c}$). We write
\begin{align}
\mathcal{J}_1 &= \mathbb{E}\left \lbrace \ln\left( XZ\right)\right \rbrace = \int_{x = 0}^{\infty} \int_{z = 0}^{\infty} \ln(xz)f_{X}(x)f_{Z}(z)\, \mathrm{d}x \, \mathrm{d}z, \nonumber
\end{align}
which can be further be written in a compact form using~\cite[4.331.1]{grad} as
\begin{equation}
\mathcal{J}_1 = -2\phi - \ln\left(\frac{1}{m_x m_z}\right),
\label{eq:J1}
\end{equation}
where $\phi$ is the Euler's constant~\cite[9.73]{grad}. We can write $\mathcal{J}_2$ as
\begin{align}
\mathcal{J}_2 = \mathbb{E}\left \lbrace \ln\left( Z+1\right)\right \rbrace = \int_{z = 0}^{\infty} \ln(z+1)f_{Z}(z)\, \mathrm{d}z,
\end{align}
which we can write using \cite[4.337.2]{grad} as
\begin{equation}
\mathcal{J}_2 = -\exp\left(\frac{1}{m_z}\right)\mathrm{Ei}\left(-\frac{1}{m_z}\right),
\label{eq:J2}
\end{equation}
where $\mathrm{Ei}(x)$ is the exponential integral~\cite[8.21]{grad}. Substituting \eqref{eq:J1} and \eqref{eq:J2} in \eqref{eq:T1}, we get the required lower bound for $T_1$.

We can rewrite $T_2$ as
\begin{align}
T_2 = \mathbb{E}\left \lbrace \ln\left(1 + \gamma_{\mathcal{R}}\right) \right \rbrace = \int_{u = 0}^{\infty} \ln(1+u)f_{\gamma_{\mathcal{R}}}(u)\, \mathrm{d}u.
\label{eq:T21}
\end{align}
Using the integration by parts method, we can rewrite \eqref{eq:T21} as
\begin{align}
T_2 = \int_{u = 0}^{\infty}\frac{1}{1+u}\left[1 -F_{\gamma_{\mathcal{R}}}(u)\right] \mathrm{d}u,
\label{eq:T22}
\end{align}
where we can write the cumulative distribution function (CDF) $F_{\gamma_{\mathcal{R}}}(u)$ as
\begin{align}
F_{\gamma_{\mathcal{R}}}(u) &= \int_{y = 0}^{\infty} F_{X}((1+y)u)f_{Y}(y)\,\mathrm{d}y \nonumber\\
&= \frac{1}{m_y}\!\int_{y = 0}^{\infty} \left[1- \exp\left(-\frac{(1+y)u}{m_x}\right)\!\right]\exp\left(\!-\frac{y}{m_y}\right)\mathrm{d}y \nonumber \\
& = 1 - \frac{m_x}{m_x + u m_y}\exp\left(-\frac{u}{m_x}\right).
\label{eq:F}
\end{align}
Substituting \eqref{eq:F} in \eqref{eq:T22} and using \cite[3.353.3]{grad} and \cite[3.352.4]{grad}, we finally obtain the required expression for $T_2$ as in \eqref{eq:T2_fin}.
\bibliographystyle{ieeetr}
\bibliography{paper}

\end{document}